\documentclass[11pt,reqno]{amsart}

\usepackage{amsmath, amsfonts, amssymb,  mathrsfs,  array, stmaryrd,  indentfirst, amsthm,  hyperref, comment}
\usepackage{graphicx, enumitem, tabularx, color}

\usepackage[margin=1.0in]{geometry}
\usepackage{setspace}
\numberwithin{equation}{section}
\theoremstyle{plain}
\newtheorem{lemma}{Lemma}[section]
\newtheorem{thm}{Theorem}[section]
\newtheorem{defn}{Definition}[section]
\newtheorem{prop}{Proposition}[section]
\newtheorem{cor}{Corollary}[section]

\newtheorem{remark}{Remark}[section]
\newtheorem{Example}{Example}[section]

\newcommand{\p}{\mathbf{p}}
\newcommand{\q}{\mathbf{q}}

\newcommand{\E}{\mathbb{E}}
\newcommand{\V}{{\text{Var}}}
\newcommand{\PP}{{\mathbb{P}}}
\newcommand{\cL}{{\mathcal{L}}}
\newcommand{\0}{{\mathbf{0}}}
\newcommand{\1}{{\mathbf{1}}}

\makeatletter
\def\@setcopyright{}
\def\serieslogo@{}
\makeatother

\begin{document}

\title[]{Time Consistent Stopping For The Mean-Standard Deviation Problem --- The Discrete Time Case}
\thanks{E. Bayraktar is supported in part by the National Science Foundation under grant DMS-1613170 and by the Susan M. Smith Professorship. We would like to thank the two anonymous referees for their incisive remarks which helped improve our presentation.}

\author{Erhan Bayraktar}
\address[Erhan Bayraktar]{Department of Mathematics, University of Michigan.}
\email{erhan@umich.edu}
\author{Jingjie Zhang }
\address[Jingjie Zhang]{Department of Mathematics, University of Michigan.}
\email{jingjiez@umich.edu}
\author{Zhou Zhou}
\address[Zhou Zhou]{School of Mathematics and Statistics, University of Sydney.}
\email{zhou.zhou@sydney.edu.au}

\begin{abstract}
Inspired by Strotz's consistent planning strategy, we formulate the infinite horizon mean-variance stopping problem as a subgame perfect Nash equilibrium in order to determine time consistent strategies with no regret.  Equilibria among stopping times or randomized stopping times may not exist. This motivates us to consider the notion of liquidation strategies, which allows the stopping right to be divisible. We then argue that 
 the mean-standard deviation variant of this problem makes more sense for this type of strategies in terms of time consistency. It turns out that an equilibrium liquidation strategy always exists. We then analyze whether optimal equilibrium liquidation strategies exist and whether they are unique and observe that neither may hold.
\end{abstract}


\maketitle
\pagestyle{headings}

\section{Introduction}
Consider an optimal stopping problem with an infinite time horizon
\begin{equation}
\sup_{\tau} \E_x[g(X_{\tau})],\label{cosp}
\end{equation}
where $X$ is a Markov process starting from state $x$, and the stopping time $\tau$ is chosen to maximize the expectation of the payoff function $g$. A classical approach to solving this optimal stopping problem is to use dynamic programming. Thanks to the particular form of (\ref{cosp}), this problem is known to be \emph{time-consistent} in the sense that its optimal stopping strategy does not depend on the initial state $x$. 
However, such property may fail to hold in some seemingly quite natural problems where the objective function is in different forms. In those cases, a stopping strategy that is optimal from ``today's'' point of view may not be optimal anymore from  ``tomorrow's'' point of view. Optimal stopping, more generally, optimal control problems with such property are said to be \emph{time-inconsistent}. 
Typical examples include non-exponential discounting, the optimization criteria used in cumulative prospect theory (e.g. rank based utility), and  the mean-variance criterion, which is the focus of our paper.

There are three ways one could deal with time-inconsistency, dating back to the seminal work of Strotz \cite{RePEc:oup:restud:v:23:y:1955:i:3:p:165-180.}. The first is to formulate an optimal stopping problem with a given initial state $x$. The optimal stopping time $\tau_x$ is then parametrized by the initial state $x$. Once the state starts at $x$, and optimal policy $\tau_x$ is determined, the player is precommitted to implementing this strategy. This strategy simply does not take the change of future preferences into account. The second is to have the agent repeatedly solve this problem, hence allow for changes in future preferences.  \cite{RePEc:oup:restud:v:23:y:1955:i:3:p:165-180.} called this strategy naive and further repercussions about this strategy are discussed in \cite{10.1257/aer.20130896}. 
\cite{MR3462068} discusses these two formulations under the labels \emph{static optimality} and \emph{dynamic optimality}.

The third way is to formulate the problem in game theoretic terms by viewing each state $x$ as a player in a game regarding when to stop the process $X$ and look for equilibrium strategies. Roughly speaking, an equilibrium strategy, which is also called a \emph{consistent plan}, can be viewed as a no-regret strategy since the agent has no incentive to deviate from the strategy at any current state $x$. 
The third way is the formulation we will follow here to analyze this problem in an infinite-horizon discrete-time setting. In discrete time,  infinite horizon is much more challenging than finite horizon because \cite{Pollak}'s backward sequential optimization approach to obtain the consistent plans no longer works.

Recently, there has been a lot of effort in determining the equilibrium strategies in stochastic control problems, see e.g. \cite{B2014} and the references therein. There are also several papers on the equilibrium strategies for stopping problems. Among them, \cite{10.1257/aer.89.1.103} analyzes the case in discrete time with a finite time horizon, and \cite{RePEc:eee:jfinec:v:84:y:2007:i:1:p:2-39} investigates a particular model in continuous time with an infinite-time horizon. A general treatment for stopping problems in continuous time is considered in \cite{MR3738666} in the context of hyperbolic discounting. In particular, \cite{MR3738666} proposes a definition of equilibrium in continuous time which avoids using the ``first order criteria'' as in control problems in the literature. \cite{MR3738666} also formulates equilibrium stopping policies as fixed points of an operator and constructs a large class of equilibria by iterating this operator. This effort is continued in \cite{2017arXiv170903535H} when the agents use probability distortions to calculate their criteria. As for the study of how to choose an equilibrium, \cite{doi:10.1137/17M1139187} considers a discrete-time infinite-horizon problem with non-exponential discounting, and investigates optimal equilibria in the sense of pointwise dominance. Apart from establishing the existence of a pure stopping equilibrium, \cite{doi:10.1137/17M1139187} also obtains the existence and uniqueness of an \emph{optimal equilibrium}. Also see \cite{2017arXiv171207806H} which is a continuous-time extension of \cite{doi:10.1137/17M1139187}. Let us also mention the recent work of \cite{doi:10.1137/17M1153029}, where the authors consider the equilibrium stopping strategies under the definition associated with first order criteria. They point out that the mean-variance problem is out of the scope of their approach. In another recent paper \cite{2018arXiv180407018C} by the same authors, a continuous time general framework for time-inconsistent stopping problems covering the mean-variance criterion is developed, and a continuous time mean-variance problem is studied, and a mixed stopping strategy for the time-inconsistent stopping problem in continuous time is defined as the first jumping time of a Cox-process associated to the state process.

In this paper, we study time-consistent mean-standard deviation stopping problems in discrete time with an infinite time horizon. We show that while a Markov equilibrium in the class of pure or randomized stopping times may not exist in general, there always exists an equilibrium liquidation strategy. In addition, we show that an optimal equilibrium in the sense of pointwise dominance may not exist, and may not be unique if it exists. We also establish the existence of a Pareto optimal equilibrium.

The main novelty in this paper is the determination of the right type of ``mixed" equilibrium strategies and the appropriate modification of the criteria to make sense of the consistent planning problem. In particular, we show that the obvious choice of mixing, i.e., a randomized stopping strategy, does not work because when an equilibrium in this class exists it coincides with the pure equilibrium strategy. (One should contrast this to the Example 2.6 of \cite{doi:10.1137/17M1153029}, where they show that there exists an equilibrium randomized stopping strategy which is not a pure stopping time.) The right notion of strategies turns out to be the \emph{liquidation strategies} that were introduced by \cite{bayraktar2016} (see also \cite{doi:10.1287/moor.2018.0932}) in the context of subhedging American options. The stopping right is taken to be divisible, or rather as a finite resource/fuel that the agent can consume continuously. The differences between randomized and liquidation strategies are highlighted in Remark~\ref{rem:diff-lq-rand}.

Instead of the mean-variance benchmark, we propose using the mean-standard deviation criterion. We think it is more meaningful in the context of consistent planning. One reason is that the mean and standard deviation are associated with the same unit. Moreover, the scaling property of this new objective function plays along well with liquidation strategies --- no matter what the history liquidation strategy is, we will face the same problem as soon as we are in the same state, because we can factor out the proportion of the stopping right remaining in the objective function. (For comparison, we also define equilibrium liquidation strategies for the mean-variance problem in a similar way. The consequent results reinforce our concerns about its properness in terms of consistent planning.)

Unlike examples in \cite{10.1257/aer.89.1.103}, the method of backward construction for equilibria fails in our setup because the time horizon is infinite. The fixed point approach used in \cite{doi:10.1137/17M1139187} (see e.g., (2.5) in \cite{doi:10.1137/17M1139187}) does not work for our problem either,  due to the non-linearity of the criterion. Instead, we provide a characterization of equilibria, from which we are able to calculate explicitly all the equilibria for many examples in this paper.

The rest of the paper is organized as follows. In Section~\ref{sec_MSD} we introduce the mean-standard deviation problem. We first analyze the equilibrium stopping time, and provide an example to show such equilibrium may not exist. Then we introduce the concepts of randomized stopping strategies and liquidation strategies and analyze the equilibria in these classes. In Section~\ref{sec_MV} we consider the similar concepts for mean-variance problems. In Section~\ref{sec_compare} we compare equilibrium liquidation strategies with statically optimal ones. Some computational details can be found in Appendix \ref{A}.

\section{Mean-Standard Deviation Problem}\label{sec_MSD}
Consider a probability space $(\Omega, \mathcal{F}, \mathbb{P})$ that supports a time-homogeneous discrete-time Markov chain $X=(X_n)_{n\in \mathbb{N}}$, taking values in a finite state space $\mathbb{X}\subset\mathbb{R}$. For each $x\in \mathbb{X}$, if $X_0=x$, we will write $X$ as $X^x$. The probability, expectation, and variance associated with $X^x$ will be denoted by $\mathbb{P}_x[\cdot]$, $\mathbb{E}_x[\cdot]$ and $\V_x[\cdot]$, respectively. We assume that the limit $X_{\infty}:=\lim_{n\to \infty} X_n$ exists almost surely. 
\subsection{Equilibrium Stopping Times}
For any $x\in \mathbb{X}$ and $\tau\in \mathcal{T}$ (where $\mathcal{T}$ is the set of stopping times w.r.t. the filtration generated by the Markov chain), consider the following objective function in mean-standard deviation problem
\begin{equation}
K_p(x, \tau):=\mathbb{E}_x[X_{\tau}]-c({\rm Var}_x[X_{\tau}])^{1/2}, \label{K_p}
\end{equation}
where $c>0$ is a constant and the subscript ``$p$'' in $K_p(x, \tau)$ stands for ``pure Markov stopping time''.

As we have discussed in the introduction part, this mean-standard deviation problem is time-inconsistent due to the non-linear term $({\rm Var}_x[X_{\tau}])^{1/2}$. We treat it as an intra-personal game regarding when to stop the process $X$ between current and future selves whose preferences, identified with the objective function $K_p$, change as the initial state $x$ changes. A reasonable equilibrium strategy should be such that once the agent chooses to follow the equilibrium strategy he will never regret no matter which state he comes into. Furthermore, we only consider the pure Markov stopping times commonly used in  game theory; see e.g. \cite{doi:10.1137/17M1153029} and \cite{MASKIN2001191}.
\begin{defn}  \label{def_psp} A stopping time $\tau$ is said to be a pure Markov stopping time, or pure stopping time for short, if  $\tau=\inf\{t\ge 0: X_t\in S\}$ for some measurable set $S\subset \mathbb{X}$ and $S$ is called the stopping region.
\end{defn}

\begin{remark}
Obviously for any stopping region $S$, the value will not change if we add or remove an absorbing state from $S$. Therefore, without loss of generality we may assume a stopping region always contains all the absorbing states. The similar argument applies to the cases when we discuss randomized stopping and liquidation strategies later on.
\end{remark}

A pure stopping time governs when the agent should stop. The decision whether to stop or not depends directly on the current state $x$ and not on the past path of process $X$. A corresponding subgame perfect Nash equilibrium based on pure Markov stopping time is defined as the following.
\begin{defn} \label{pure eq} A pure Markov stopping time $\tau$ with stopping region $S$ is said to be an equilibrium stopping time for \eqref{K_p} if
\begin{equation}\label{ee7}
x\ge K_p(x, \rho(x, S)),\,\, \forall x\in S \quad \quad {\rm and}\quad \quad x\leq K_p(x, \rho(x, S)),\,\, \forall x\notin S, 
\end{equation}
where $\rho(x, S):=\inf\{n\ge 1: X^x_n\in S\}\in \mathcal{T}.$ 
\end{defn}
The next result shows that an equilibrium stopping time may not exist.
\begin{prop} \label{eg013610}An equilibrium stopping time does not always exist.
\end{prop}
\begin{proof} We will prove this by giving a counterexample.
Let $c=1$. $X$ has state space $\mathbb{X}=\{0,1,3,6,10\}$ and the following transition matrix.\[
\begin{matrix}
 &\mathbf{0} & \mathbf{1}  & \mathbf{3} & \mathbf{6} & \mathbf{10}\\
\mathbf{0}& 1 & 0 & 0 & 0 &0\\
\mathbf{1}& 0.2 & 0 & 0.4 & 0.2 &0.2\\
\mathbf{3}& 0 & 0 & 1 & 0 &0\\
\mathbf{6}& 0 & 0.2 & 0 & 0 &0.8\\
\mathbf{10}& 0 & 0 & 0 & 0 &1\\
  \end{matrix}
\]

For this Markov chain, $\{0,3,10\}$ are absorbing states and $\{1,6\}$ are transient. Suppose there exists an equilibrium stopping time with stopping region $S\subset \{0,1,3,6,10\}$ and consider the following four cases.

\textbf{Case 1:}
$S=\{0,1,3,6,10\}$. We have that \[
\mathbb{P}_1(X_{\rho(1, S)}=0)=\mathbb{P}_1(X_{\rho(1,S)}=6)=\mathbb{P}_1(X_{\rho(1,S)}=10)=0.2, \,\,\,\,  \mathbb{P}_1(X_{\rho(1,S)}=3)=0.4,
\]
then \[
\mathbb{E}_1[X_{\rho(1,S)}]=4.4, \,\, \mathbb{E}_1[X_{\rho(1,S)}^2]=30.8 \Rightarrow K_p(1,\rho(1,S))=1.0177>1,
\]
which yields a contradiction.

\textbf{Case 2:}
$S=\{0,3,6,10\}$. We have that \[
\mathbb{P}_6(X_{\rho(6,S)}=0)=\mathbb{P}_6(X_{\rho(6,S)}=6)=0.04, \,\,\,\, \mathbb{P}_6(X_{\rho(6,S)}=3)=0.08, \,\,\,\, \mathbb{P}_6(X_{\rho(6,S)}=10)=0.84,
\]
then \[
\mathbb{E}_6[X_{\rho(6,S)}]=8.88,\,\,
 \mathbb{E}_6[X_{\rho(6,S)}^2]=86.16 \Rightarrow K_p(6,\rho(6,S))=6.1771>6,
\]
which yields a contradiction.

\textbf{Case 3:}
$S=\{0,1,3,10\}$. We have that \[
\mathbb{P}_6(X_{\rho(6,S)}=1)=0.2, \,\,\,\, \mathbb{P}_6(X_{\rho(6,S)}=10)=0.8,
\]
then \[
\mathbb{E}_6[X_{\rho(6,S)}]=8.2,\,\, \mathbb{E}_6[X_{\rho(6,S)}^2]=80.2 \Rightarrow K_p(6,\rho(6,S))=4.6<6,
\]
which yields a contradiction.

\textbf{Case 4:}
$S=\{0,3,10\}$. We have that \[
\mathbb{P}_1(X_{\rho(1,S)}=0)=\frac{5}{24}, \,\,\,\, \mathbb{P}_1(X_{\rho(1,S)}=3)=\frac{5}{12} ,\,\,\,\, \mathbb{P}_1(X_{\rho(1,S)}=10)=\frac{3}{8},
\]
then  \[
\mathbb{E}_1[X_{\rho(1,S)}]=5, \,\,\mathbb{E}_1[X_{\rho(1,S)}^2]=\frac{165}{4} \Rightarrow K_p(1,\rho(1,S))=0.9689<1,
\]
which yields a contradiction.
\end{proof}

\subsection{Randomized Stopping Times}

In the last section, we observed that there is no guarantee that an equilibrium stopping time exists. Therefore we will now seek an equilibrium among randomized stopping times. 

Let us first briefly recall some facts of randomized stopping times, and we refer to \cite{MR668536} for more details.  A randomized stopping time (w.r.t. the original space $(\Omega,\mathcal{F})$) is defined to be a stopping time w.r.t. the extended space $(\Omega\times[0,1], \mathcal{F}\otimes\mathcal{B}([0,1]))$. For a randomized stopping time $\gamma:\,\Omega\times[0,1]\mapsto\mathbb{N}$, its $\omega$-distribution is defined by
$$M_k(\omega):=\text{Leb}\{v:\ \gamma(\omega,v)\leq k\},\quad k\in\mathbb{N},\ \omega\in\Omega,$$
where $\text{Leb}$ is the Lebesgue measure. Intuitively, $M_k(\omega)$ represents the probability that the underlying process has stopped by time $k$ along the path $\omega$. Moreover, there is a one-to-one correspondence (up to a rearrangement) between $\gamma$ and $M$; see \cite{MR668536}. For the Markov chain $X$ and any randomized stopping time $\gamma$ with the $\omega$-distribution $M$, denote
$$\E[X_\gamma]=\E\left[M_0 X_0+\sum_{k=1}^\infty X_k(M_k-M_{k-1})+(1-M_\infty)X_\infty\right],$$
and
$${\rm Var}[X_{\gamma}]=\E[X_\gamma^2]-(\E[X_\gamma])^2,$$
 $M_\infty:=\lim_{n\rightarrow\infty}M_n$. 

\begin{defn} We say $\gamma$ is a time-homogeneous randomized stopping time, if there exists $\mathbf{p}: \mathbb{X}\to [0,1]$, such that the $\omega$-distribution of $\gamma$ satisfies
$$M_k(\cdot)=1-\prod_{n=0}^k(1-\mathbf{p}(X_k(\cdot))).$$
Here $\mathbf{p}(x)$ represents the probability to stop at state $x$, given the underlying process has not stopped yet.
We call $\mathbf{p}: \mathbb{X}\to [0,1]$ a randomized stopping strategy, and denote the set of all of them by $\mathcal{P}$.
\end{defn}

Intuitively, given a function $\mathbf{p}: \mathbb{X}\to [0,1]$, we can design $n$ biased coins, where $n$ is the number of states in $\mathbb{X}$. When we are at state $X_k = x$, we will flip the coin with probability $\mathbf{p}(x)$ it comes up heads. If it comes up heads, we will stop. Otherwise, we will continue. In general, we can design more complicated strategies about flipping coins, which will fit in with general randomized stopping times, not just time-homogeneous randomized stopping times.

For any $\mathbf{p},\mathbf{q}\in\mathcal{P}$, denote $\gamma^{\mathbf{q}\otimes\mathbf{p}}$ as the randomized stopping time with the $\omega$-distribution
$$M_0=\mathbf{q}(X_0),\quad\text{and}\quad M_k=1-(1-\mathbf{q}(X_0))\prod_{n=1}^k(1-\mathbf{p}(X_k)),\quad k=1,2,\dotso.$$
We sometimes also write $\gamma^\p$ instead of $\gamma^{\p\otimes\p}$ for short. With a bit abuse of notation, we use $\E[X_{\mathbf{q}\otimes\mathbf{p}}]$  to represent $\E[X_{\gamma^{\mathbf{q}\otimes\mathbf{p}}}]$, and ${\rm Var}[X_{\mathbf{q}\otimes\mathbf{p}}]$ to represent ${\rm Var}[X_{\gamma^{\mathbf{q}\otimes\mathbf{p}}}]$.

In Definition \ref{pure eq}, an equilibrium stopping time is a subgame perfect Nash equilibrium in which all players use pure Markov stopping times. Now we propose to consider an equilibrium randomized stopping strategy which is a subgame perfect Nash equilibrium in the game where all players use time-homogeneous randomized stopping times and their preferences are identified with the following objective function
\begin{equation}
K_r(x, \mathbf{p}):=\mathbb{E}_x[X_{\mathbf{p}}]-c({\rm Var}_x[X_{\mathbf{p}}])^{1/2}, \label{K_r}
\end{equation}
where $\p$ is a randomized stopping strategy and the subscript ``$r$'' in $K_r$ stands for ``randomized stopping strategy''.

\begin{defn} \label{mixed eqr} 
$\mathbf{p}\in \mathcal{P}$ is said to be an equilibrium randomized stopping strategy for (\ref{K_r}), if for any mapping $\mathbf{q}: \mathbb{X}\to [0,1]$,
\begin{equation}\label{ee2}
K_r(x, \mathbf{q}\otimes \mathbf{p})\le K_r(x, \p\otimes\mathbf{p}), \,\,\, \forall x\in \mathbb{X},
\end{equation}
where $K_r(x, \mathbf{q}\otimes \mathbf{p})$ is from \eqref{K_r} by replacing $\p$ with $\p\otimes\q$.
\end{defn}

The randomized strategy is also called ``mixed strategy'' in game theory, i.e., an assignment of a probability to each pure strategy. In our context, we assign probability $\mathbf{p}(x)$ to the pure strategy ``to stop at state $x$'' and probability $1-\mathbf{p}(x)$ to the pure strategy ``not to stop at state $x$''.

If the randomized stopping strategy $\p\in\mathcal{P}$ satisfies that $\p(x) \in \{0, 1\}$ for any $x\in\mathbb{X}$, then it is actually a pure strategy, which is to stop at state $x$ if $\p(x)=1$ and not to stop at state $x$  if $\p(x)=0$. In this case, it simply gives us a pure stopping time with stopping region $\{x\in\mathbb{X}:\ \p(x)=1\}$. We have the following result, which together with Proposition \ref{eg013610}  implies that an equilibrium randomized stopping strategy does not always exist.

\begin{prop} \label{pp1}
If $\p\in\mathcal{P}$ is an equilibrium randomized stopping strategy for (\ref{K_r}), then $\p(x)=0$ or $1$ for any (transient state) $x\in\mathbb{X}$. Conversely, if there is an equilibrium stopping time with stopping region $S$, then $\p\in\mathcal{P}$ defined by
\begin{equation}\label{ee5}
\p(x)=
\begin{cases}
1,& x\in S,\\
0,&x\notin S,
\end{cases}
\end{equation}
is an equilibrium randomized stopping strategy. Consequently, an equilibrium randomized stopping strategy does not always exist.
\end{prop}

For the proof of this proposition we will need the following result:

\begin{lemma}\label{ll1}
Let $\gamma_1,\gamma_2,\gamma$ be randomized stopping times and $\lambda\in(0,1)$, such that
$$\PP(\gamma=\gamma_1)=1-\PP(\gamma=\gamma_2)=\lambda.$$
(Denote $\gamma=\lambda\gamma_1\oplus(1-\lambda)\gamma_2$.) Then
$${\rm Var}[X_{\lambda\gamma_1\oplus(1-\lambda)\gamma_2}]\geq \lambda{\rm Var}[X_{\gamma_1}]+(1-\lambda){\rm Var}[X_{\gamma_2}]\geq\left(\lambda({\rm Var}[X_{\gamma_1}])^{1/2}+(1-\lambda)({\rm Var}[X_{\gamma_2}])^{1/2}\right)^2.$$
Moreover, the first equality holds if and only if $\E[X_{\gamma_1}]=\E[X_{\gamma_2}]$, and the second equality holds if and only if ${\rm Var}[X_{\gamma_1}]={\rm Var}[X_{\gamma_2}]$.
\end{lemma}
\begin{proof}
We have that
\begin{align*}
{\rm Var}[X_{\lambda\gamma_1\oplus(1-\lambda)\gamma_2}]&=\mathbb{E}[X_{\lambda\gamma_1\oplus(1-\lambda)\gamma_2}^2]-(\mathbb{E}[X_{\lambda\gamma_1\oplus(1-\lambda)\gamma_2}])^2\\
&=\lambda\mathbb{E}[X_{\gamma_1}^2]+(1-\lambda)\mathbb{E}[X_{\gamma_2}^2]-\left(\lambda\mathbb{E}[X_{\gamma_1}]+(1-\lambda)\mathbb{E}[X_{\gamma_2}]\right)^2\\
&\geq\lambda\mathbb{E}[X_{\gamma_1}^2]+(1-\lambda)\mathbb{E}[X_{\gamma_2}^2]-\left(\lambda(\mathbb{E}[X_{\gamma_1}])^2+(1-\lambda)(\mathbb{E}[X_{\gamma_2}])^2\right)\\
&=\lambda{\rm Var}[X_{\gamma_1}]+(1-\lambda){\rm Var}[X_{\gamma_2}].
\end{align*}
 We obtain the inequality using Jensen's inequality.
The rest of the result is easy to check.
\end{proof}

\begin{proof}[Proof of Proposition \ref{pp1}]
Let $\p\in\mathcal{P}$ be an equilibrium randomized stopping strategy for (\ref{K_r}). Suppose there exists a transient state $x\in\mathbb{X}$ such that $0<\lambda:=\p(x)<1$. Denote
\begin{equation}\label{ee6}
\alpha:=\1\otimes\p\quad\text{and}\quad\beta:=\0\otimes\p,
\end{equation}
where $\1\in\mathcal{P}$ (resp. $\0\in\mathcal{P}$) is the strategy with all components $1$ (resp. $0$).  We have the following.
\begin{align}
\notag K_r(x,\p\otimes\p)&=K_r(x,\lambda\alpha\oplus(1-\lambda)\beta)\\
\notag &=\E\left[X_{\lambda\alpha\oplus(1-\lambda)\beta}\right]-c\left({\rm Var}[X_{\lambda\alpha\oplus(1-\lambda)\beta}]\right)^{1/2}\\
\label{ee3} &\leq\lambda\E[X_\alpha]+(1-\lambda)\E[X_\beta]-c\left(\lambda\left({\rm Var}[X_\alpha]\right)^{1/2}+(1-\lambda)\left({\rm Var}[X_\beta]\right)^{1/2}\right)\\
\notag &=\lambda K_r(x,\alpha)+(1-\lambda) K_r(x,\beta)\\
\label{ee4} &\leq K_r(x,\p\otimes\p),
\end{align}
where \eqref{ee3} follows from Lemma \ref{ll1} and \eqref{ee4} follows from \eqref{ee2}. This implies that equality holds for \eqref{ee3}. By Lemma \ref{ll1}
$$x=X_\alpha=X_\beta.$$
Since the state $x$ is transient, there is a positive probability that the Markov chain never returns back to $x$. As a result, it is not possible that $X_\beta=x$ with probability $1$. 

Conversely, assume there is an equilibrium stopping time with stopping region $S$, and define $\p\in\mathcal{P}$ as in \eqref{ee5}. Let $\q\in\mathcal{P}$ and $x\in\mathbb{X}$. Denote $\lambda':=\q(x)$, and define $\alpha$ and $\beta$ as in \eqref{ee6}. We consider two cases:\\
(i) $p(x)=1$: Then
$$K_r(x,\beta)=K_p(x,\rho(x,S))\leq x=K_r(x,\alpha).$$
Then by a similar argument as above, we have that
$$K_r(x,\q\otimes\p)\leq \lambda' K_r(x,\alpha)+(1-\lambda')K_r(x,\beta)\leq K_r(x,\alpha)=K_r(x,\p\otimes\p).$$\\
(ii) $p(x)=0$: $K_r(x, \beta)=K_r(x, \p\otimes \p)\ge K_r(x, \alpha)=x$, and thus $K_r(x, \q\otimes \p)\le K_r(x,\beta)=K_r(x, \p \otimes \p)$.
\end{proof}

\subsection{Equilibrium Liquidation Strategies}

\begin{defn}
An adapted nondecreasing process $\theta=(\theta_n)_{n\in\mathbb{N}}$ is said to be a liquidation strategy, if $\theta_0\geq 0$, and
$$\lim_{n\rightarrow\infty}\theta_n\leq 1,\ a.s..$$
A liquidation strategy $\theta$ is said to be time homogeneous, if there exists $\eta:\,\mathbb{X}\mapsto[0,1]$, such that along any path $(x_n)_{n\in\mathbb{N}}\in\mathbb{X}^\infty$,
\begin{equation*}
\theta_n(x_0,\dotso,x_n)=1-\prod_{i=0}^n(1-\eta(x_i)) \label{theta(X)}.
\end{equation*}
Denote by $\mathcal{L}$ the collection of all time-homogeneous liquidation strategies.
\end{defn}

Consider the objective function 
\begin{equation}
K_l(x,\theta):=\mathbb{E}_x[\theta(X)]-c({\rm Var}_x[\theta(X)])^{1/2}, \label{K_l}
\end{equation}
where the subscript ``$l$'' in $K_l(x, \theta)$ stands for ``liquidation strategy'' and $\theta(X)$ is the payoff under liquidation strategy $\theta$ 
\begin{equation*}
\theta(X)=X_0\theta_0+\sum_{n=1}^{\infty}X_n(\theta_n-\theta_{n-1})+X_{\infty}(1-\theta_{\infty}).
\end{equation*}
If $\theta=\theta^{\eta}\in \mathcal{L}$ is a time-homogeneous liquidation strategy, then
\begin{equation}\label{thetaX}
\begin{split}
\theta^{\eta}(X)=&\,\eta(X_0)X_0+(1-\eta(X_0))\bigg[\eta(X_1)X_1\\
&+\sum_{k=2}^{\infty}(1-\eta(X_1))\cdots(1-\eta(X_{k-1}))\eta(X_k)X_k+\prod_{k=1}^{\infty}(1-\eta(X_k))X_{\infty}\bigg].
\end{split}
\end{equation}

Intuitively liquidation strategy means to liquidate the asset at several periods instead of at one time. Such strategy is very common in practice. For instance, when an investor has a large amount of identical asset, e.g., 10000 shares of  American option, she may excise these shares at different times instead of once. In the following, we use an example to illustrate the motivation to consider liquidation strategies. In particular, we will show that if $X$ is divisible, then it is possible that the optimal value for pure stopping time $\sup_{\tau}K_p(x, \tau)$ is strictly less than $K_l(x, \theta)$ for some liquidation strategy $\theta$ and some $x\in \mathbb{X}$.

\begin{Example} \label{eg0123}
Let $c=1/(\sqrt{44}-5)$. $X$ has the following transition matrix.
$$
\begin{array}{ccccc}
&{\bf0}&{\bf1}&{\bf2}&{\bf3}\\
{\bf0}&1&0&0&0\\
{\bf1}&\frac{1}{6}&0&\frac{1}{2}&\frac{1}{3}\\
{\bf2}&\frac{1}{5}&0&0&\frac{4}{5}\\
{\bf3}&0&0&0&1
\end{array}
$$
Then it is easy to see that the optimal stopping value for $K_p(x, \tau)$ is given by
$$\sup_\tau K_p(1,\tau)=K_p(1,\tau')=K_p(1,\tau'')=2-c=1.3877,$$
where
$$\tau':=\inf\{n\geq 0:\ X_n=0,2,3\},$$
and
$$\tau'':=\inf\{n\geq 0:\ X_n=0,3\}.$$
Now consider the liquidation strategy $\theta'$ given by
$$\theta_0'(\cdot)=0,\quad\theta_1'(\cdot,2)=1/2,\quad\text{and}\ \theta_n'=0\ \text{for all other cases}.$$
Then it is easy to see that
$$\theta'(X^1)=\frac{1}{2}X_{\tau'}^1+\frac{1}{2}X_{\tau''}^1.$$
Then the distribution of $\theta(X^1)$ is given by
$$\PP(\theta'(X^1)=0)=\frac{1}{6},\quad\PP(\theta'(X^1)=1)=\frac{1}{10},\quad\PP(\theta'(X^1)=5/2)=\frac{2}{5},\quad\PP(\theta'(X^1)=3)=\frac{1}{3}.$$
Therefore, we have that
$$\sup_\theta K_l(1,\theta)\geq K_l(1,\theta')=\frac{21}{10}-\frac{\sqrt{119}}{10}c=1.4321>\sup_\tau K_p(1,\tau).$$
\end{Example}

\begin{remark}\label{rem:diff-lq-rand}
As discussed in Remark 3.1 in Bayraktar and Zhou's paper \cite{bayraktar2016}, there is a one-to-one correspondence between the set of time-homogeneous liquidation strategies $\mathcal{L}$ and the set of time-homogeneous randomized stopping times $\mathcal{P}$. But the paths of a liquidation strategy and a randomized stopping time are quite different. First of all, in terms of behavior, when using a  randomized stopping time, we flip a coin at each period to decide whether we stop or not, and we still liquidate the whole unit asset over a single period. Second, in terms of variance, randomized stopping time will result in a larger variance, since the overall variance will include the part from randomization of the stopping time, while liquidation strategy results in a smaller variance, since averaging random variable leads to a smaller variance. This can also be seen from Lemma \ref{ll1} and the following result.
\end{remark}

\begin{lemma} \label{lemma0}
Let $\theta_1,\theta_2$ be two liquidation strategies and $\lambda\in(0,1)$. Then $\lambda\theta_1+(1-\lambda)\theta_2$ is also a liquidation strategy and
\begin{align}
\notag{\rm Var}[(\lambda\theta_1+(1-\lambda)\theta_2)(X)]&\leq\left(\lambda({\rm Var}[\theta_1(X)])^{1/2}+(1-\lambda)({\rm Var}[\theta_2(X)])^{1/2}\right)^2\\
\notag&\leq\lambda{\rm Var}[\theta_1(X)]+(1-\lambda){\rm Var}[\theta_2(X)].
\end{align}
\end{lemma}

\begin{proof}
We have that \begin{align*}
&{\rm Var}[(\lambda\theta_1+(1-\lambda)\theta_2)(X)]={\rm Var}[\lambda\theta_1(X)+(1-\lambda)\theta_2(X)]\\
&=\lambda^2{\rm Var}[\theta_1(X)]+(1-\lambda)^2{\rm Var}[\theta_2(X)]+2\lambda(1-\lambda)\text{Cov}[\theta_1(X),\theta_2(X)]\\
&\leq\lambda^2{\rm Var}[\theta_1(X)]+(1-\lambda)^2{\rm Var}[\theta_2(X)]+2\lambda(1-\lambda)({\rm Var}[\theta_1(X)])^{1/2}({\rm Var}[\theta_1(X)])^{1/2}\\
&=\left(\lambda({\rm Var}[\theta_1(X)])^{1/2}+(1-\lambda)({\rm Var}[\theta_2(X)])^{1/2}\right)^2.
\end{align*}
The second inequality is easy to check.
\end{proof}
Our next goal is to analyze the subgame perfect Nash equilibrium in the game where all players use time-homogeneous liquidation strategies. Notice that each time-homogeneous liquidation strategy is characterized by a function $\eta(x)$ that represents the proportion of the remaining asset we will liquidate when the Markov chain moves to position $x$. $\eta(x)$ is independent of time and the history of the paths. For simplicity of notation, we use $K_l(x, \eta)$ instead of $K_l(x, \theta^{\eta})$ for $\theta=\theta^{\eta}\in \mathcal{L}$.

\begin{defn} \label{eql} A liquidation strategy $\theta=\theta^{\eta}\in \mathcal{L}$ is said to be an equilibrium liquidation strategy for (\ref{K_l})
if for any mapping $\xi: \mathbb{X}\to [0,1]$, we have
\[K_l(x, {\xi\otimes\eta})\le K_l(x, {\eta\otimes \eta}),\,\,\, \forall x\in \mathbb{X},\]
where $\theta^{\xi\otimes \eta}$ is a perturbation of strategy $\theta^{\eta}$ in which we liquidate $\xi(\cdot)$ at time 0 and then from time 1 we liquidate $\eta(\cdot)$ proportion of the remaining asset at each period. 
\end{defn}
\begin{remark} \label{rmk1} Notice that this definition looks similar to the definition of equilibrium randomized stopping time. But as we mentioned earlier, unlike selling the whole unit of asset in one period which is random, the liquidation strategy will leave us with different proportions of asset at different periods, so the objective function might change as time goes on. Thanks to the square root term in (\ref{K_l}), mean-standard deviation problem has the scaling effect which allows Definition \ref{eql} to make perfect sense since we essentially face the same problem (with the same parameter $c$) at every period. However, we will see in Section \ref{sec_MV}, a similar definition of equilibrium liquidation strategy in mean-variance problem is not a proper definition.

Also note that a liquidation strategy should be considered as a pure strategy from the game theory point of view, with the added assumption that partial selling the asset over time is possible. In contrast, a randomized stopping strategy should be considered as a mixed strategy. 
\end{remark}

\subsection{Existence of an Equilibrium Liquidation Strategy}
In this section we will prove that in contrast to the equilibrium stopping time for (\ref{K_p}) and equilibrium randomized stopping strategy for (\ref{K_r}), an equilibrium liquidation strategy for (\ref{K_l}) always exists.
\begin{lemma} \label{lemma1}
For $\eta_n,\eta\in\mathcal{L}, n\in \mathbb{N}$, if $\eta_n\rightarrow\eta$ as $n\to \infty$, then
\[\theta^{\eta_n}(X)\rightarrow\theta^\eta(X),\quad\text{a.s.}.\]
\end{lemma}
\begin{proof}
For a.e. $\omega\in \Omega$, there exists $N=N(\omega)$ such that for any $k\geq N$, $X_N(\omega)=X_\infty(\omega)$. Then along $\omega$, we have that
\begin{align*}
\theta^{\eta_n}(X)&=\sum_{k=0}^\infty\left((1-\eta_n(X_0))\dotso(1-\eta_n(X_{k-1}))\right)\eta_n(X_k)X_k+\prod_{k=0}^\infty(1-\eta_n(X_k))X_\infty\\
&=\sum_{k=0}^N\left((1-\eta_n(X_0))\dotso(1-\eta_n(X_{k-1}))\right)\eta_n(X_k)X_k+\prod_{k=0}^N(1-\eta_n(X_k))X_\infty\\
&\xrightarrow{n\to \infty}\sum_{k=0}^N\left((1-\eta(X_0))\dotso(1-\eta(X_{k-1}))\right)\eta(X_k)X_k+\prod_{k=0}^N(1-\eta(X_k))X_\infty\\
&=\sum_{k=0}^\infty\left((1-\eta(X_0))\dotso(1-\eta(X_{k-1}))\right)\eta(X_k)X_k+\prod_{k=0}^\infty(1-\eta(X_k))X_\infty\\
&=\theta^\eta(X).
\end{align*}
\end{proof}
\begin{lemma}\label{lemma2}
For $\xi_n,\eta_n,\xi,\eta \in\mathcal{L}, n\in \mathbb{N}$, if $\xi_n\rightarrow\xi$ and $\eta_n\rightarrow\eta$ as $n\to\infty$, then
$$K_l(x,\xi_n\otimes\eta_n)\rightarrow K_l(x,\xi\otimes\eta),\quad\forall\,x\in\mathbb{X}.$$
\end{lemma}
\begin{proof}
As
$$|K_l(x,\xi_n\otimes\eta_n)-K_l(x,\xi\otimes\eta)|\leq|K_l(x,\xi_n\otimes\eta_n)-K_l(x,\xi\otimes\eta_n)|+|K_l(x,\xi\otimes\eta_n)-K_l(x,\xi\otimes\eta)|,$$
it suffices to show that
\begin{align}
\label{e7}&\left|\E_x\left[\theta^{\xi_n\otimes\eta_n}(X)\right]-\E_x\left[\theta^{\xi\otimes\eta_n}(X)\right]\right| \rightarrow 0, \quad n\to \infty;\\
\label{e8}&\left|\E_x\left[\left(\theta^{\xi_n\otimes\eta_n}(X)\right)^2\right]-\E_x\left[\left(\theta^{\xi\otimes\eta_n}(X)\right)^2\right]\right| \rightarrow 0, \quad n\to \infty;\\
\label{e9}&\left|\E_x\left[\theta^{\xi\otimes\eta_n}(X)\right]-\E_x\left[\theta^{\xi\otimes\eta}(X)\right]\right| \rightarrow 0, \quad n\to \infty;\\
\label{e10}&\left|\E_x\left[\left(\theta^{\xi\otimes\eta_n}(X)\right)^2\right]-\E_x\left[\left(\theta^{\xi\otimes\eta}(X)\right)^2\right]\right| \rightarrow 0,  \quad n\to \infty.
\end{align}
We have that
\begin{align}
\notag&\left|\E_x\left[\theta^{\xi_n\otimes\eta_n}(X)\right]-\E_x\left[\theta^{\xi\otimes\eta_n}(X)\right]\right|\\
\notag=&\left|\left(\xi_n(x)x+(1-\xi_n(x))\sum_{y\in\mathbb{X}}p(x,y)\E_y\left[\theta^{\eta_n}(X)\right]\right)-\left(\xi(x)x+(1-\xi(x))\sum_{y\in\mathbb{X}}p(x,y)\E_y\left[\theta^{\eta_n}(X)\right]\right)\right|\\
\notag\leq&|x|\cdot|\xi_n(x)-\xi(x)|+\left|\sum_{y\in\mathbb{X}}p(x,y)\E_y\left[\theta^{\eta_n}(X)\right]\right|\cdot|\xi_n(x)-\xi(x)|\\
\notag\leq& (\alpha+|x|)|\xi_n(x)-\xi(x)|
\notag \rightarrow 0,  \text{as} \,\, n\to \infty,
\end{align}
where $\alpha:=\sup\{|y|:\ y\in\mathbb{X}\}$. Hence, we have \eqref{e7} holds.

Noticing that
$$\theta^{\xi_n\otimes\eta_n}(X)=\xi_n(x)x+(1-\xi_n(x))\theta^{\eta_n}(X_{\cdot+1}),$$
we have that
\begin{align}
\notag&\quad\, \left|\E_x\left[\left(\theta^{\xi_n\otimes\eta_n}(X)\right)^2\right]-\E_x\left[\left(\theta^{\xi\otimes\eta_n}(X)\right)^2\right]\right|\\
\notag&=\Big|\E_x\left[(\xi_n(x))^2x^2+2\xi_n(x)x(1-\xi_n(x))\theta^{\eta_n}(X_{\cdot+1})+(1-\xi_n(x))^2(\theta^{\eta_n}(X_{\cdot+1}))^2\right]\\
\notag&\quad\quad-\E_x\left[(\xi(x))^2x^2+2\xi(x)x(1-\xi(x))\theta^{\eta_n}(X_{\cdot+1})+(1-\xi(x))^2(\theta^{\eta_n}(X_{\cdot+1}))^2\right]\Big|\\
\notag&\leq\left|(\xi_n(x))^2x^2-(\xi(x))^2x^2\right|+2|\xi_n(x)x(1-\xi_n(x))-\xi(x)x(1-\xi(x))|\cdot\left|\E_x\left[\theta^{\eta_n}(X_{\cdot+1})\right]\right|\\
\notag&\quad\quad+\left|(1-\xi_n(x))^2-(1-\xi(x))^2\right|\cdot\left|\E_x\left[(\theta^{\eta_n}(X_{\cdot+1}))^2\right]\right|\\
\notag&\leq\left|(\xi_n(x))^2x^2-(\xi(x))^2x^2\right|+2|\xi_n(x)x(1-\xi_n(x))-\xi(x)x(1-\xi(x))|\cdot\alpha\\
\notag&\quad\quad+\left|(1-\xi_n(x))^2-(1-\xi(x))^2\right|\cdot\alpha^2
\notag \rightarrow 0 ,  \text{as} \,\, n\to \infty,
\end{align}
and thus \eqref{e8} follows.

Moreover,
\begin{align}
\notag&\left|\E_x\left[\theta^{\xi\otimes\eta_n}(X)\right]-\E_x\left[\theta^{\xi\otimes\eta}(X)\right]\right|\\
\notag=&\left|\left(\xi(x)x+(1-\xi(x))\sum_{y\in\mathbb{X}}p(x,y)\E_y\left[\theta^{\eta_n}(X)\right]\right)-\left(\xi(x)x+(1-\xi(x))\sum_{y\in\mathbb{X}}p(x,y)\E_y\left[\theta^{\eta}(X)\right]\right)\right|\\
\notag\leq&(1-\xi(x))\sum_{y\in\mathbb{X}}p(x,y)\left|\E_y\left[\theta^{\eta_n}(X)\right]-\E_y\left[\theta^{\eta}(X)\right]\right|
\notag \rightarrow 0,  \text{as} \,\, n\to \infty,
\end{align}
where the last line follows from Lemma \ref{lemma1}, and thus \eqref{e9} holds.

Finally,
\begin{align}
\notag&\quad\,\left|\E_x\left[\left(\theta^{\xi\otimes\eta_n}(X)\right)^2\right]-\E_x\left[\left(\theta^{\xi\otimes\eta}(X)\right)^2\right]\right|\\
\notag&=\Big|\E_x\left[(\xi(x))^2x^2+2\xi(x)x(1-\xi(x))\theta^{\eta_n}(X_{\cdot+1})+(1-\xi(x))^2(\theta^{\eta_n}(X_{\cdot+1}))^2\right]\\
\notag&\quad\quad-\E_x\left[(\xi(x))^2x^2+2\xi(x)x(1-\xi(x))\theta^{\eta}(X_{\cdot+1})+(1-\xi(x))^2(\theta^{\eta}(X_{\cdot+1}))^2\right]\Big|\\
\notag&\leq2\xi(x)|x|(1-\xi(x))\left|\E_x\left[\theta^{\eta_n}(X_{\cdot+1})\right]-\E_x\left[\theta^{\eta}(X_{\cdot+1})\right]\right|\\
\notag&\quad\quad+(1-\xi(x))^2\left|\E_x\left[(\theta^{\eta_n}(X_{\cdot+1}))^2\right]-\E_x\left[(\theta^{\eta}(X_{\cdot+1}))^2\right]\right|
\notag \rightarrow 0,  \text{as} \,\, n\to \infty,
\end{align}
which implies \eqref{e10}.
\end{proof}

\begin{thm} \label{thm1}
There exists an equilibrium liquidation strategy for the mean-standard deviation problem (\ref{K_l}).
\end{thm}
\begin{proof}
For $\eta\in\cL$, define the set valued map
$$\Phi(\eta):=\{\xi^*\in\cL:\ K_l(x,\xi^*\otimes\eta)\geq K_l(x,\xi\otimes\eta),\ \forall\,x\in\mathbb{X}, \forall \, \xi\in \cL\}.$$

First, we have that for any $\eta\in\cL$, $\Phi(\eta)$ is not empty. Indeed, since $K_l(x,\xi\otimes\eta)$ depends on $\xi$ only through $\xi(x)$, we can choose $\xi^*(x)$ to be a maximizer for each $x$ fixed. 

For $\xi_1,\xi_2,\eta\in\cL$ and $\lambda\in(0,1)$, we have that $\theta^{\eta\otimes\eta}=\theta^\eta$ and
\begin{equation*}
\theta^{(\lambda\xi_1+(1-\lambda)\xi_2)\otimes\eta}=\lambda\theta^{\xi_1\otimes\eta}+(1-\lambda)\theta^{\xi_2\otimes\eta}.
\end{equation*}
Moreover, thanks to Lemma \ref{lemma0}, we obtain that $\Phi(\eta)$ is a convex set for any $\eta\in\cL$. In addition, by Lemma \ref{lemma2}, the map $\Phi$ is u.s.c.. That is, for $\eta_n,\xi_n^*,\eta,\xi^*\in\cL$ with $\eta_n\rightarrow\eta$ and $\xi_n^*\rightarrow\xi^*$, if $\xi_n^*\in\Phi(\eta_n)$, then $\xi^*\in\Phi(\eta)$.

Applying \cite[Theorem 1]{MR0047317}, we obtain the desired result.
\end{proof}

\begin{remark}
As we can see in this proof, the assumption that $X$ have finite state space and the limit $X_{\infty}$ exists a.s. is necessary to obtain some key estimations which are hard to achieve when the state space is infinite. Despite this, the concepts introduced in this paper do not rely on this assumption and a future work can focus on extending certain results to the case with infinite state space.
\end{remark}

\subsection{Optimal Equilibrium Liquidation Strategies}

According to consistent planning in Strotz \cite{RePEc:oup:restud:v:23:y:1955:i:3:p:165-180.}, finding equilibria is only the first step and the agent should choose the best one among all equilibria. We then formulate the definition of optimal equilibrium liquidation strategy as the following. 

\begin{defn}\label{om} Let $\mathcal{E}\subset \mathcal{L}$ be the collection of equilibrium liquidation strategies. We say an equilibrium liquidation strategy $\eta^*\in \mathcal{E}$ is optimal if
\[
K_l(x, {\eta^*})\ge K_l(x, {\eta}), \quad  
\forall \, x\in \mathbb{X}, \forall \, \eta\in \mathcal{E}.
\]
\end{defn}

To find an optimal equilibrium liquidation strategy, we need to study the set $\mathcal{E}$. We will provide a characterization of equilibrium liquidation strategies in the following proposition.

\begin{prop} \label{lemma3}
$\eta$ is an equilibrium liquidation strategy if and only if the following holds:

(i) $\eta(x)=0$ for all $x\in \mathbb{X}$ such that  $\mathbb{E}_x[Y_{\eta}]-c({\rm Var}_x[Y_{\eta}])^{1/2}>x$, and 

(ii) $\eta(x)=1$ for all $x\in \mathbb{X}$ such that  $\mathbb{E}_x[Y_{\eta}]-c({\rm Var}_x[Y_{\eta}])^{1/2}<x$,

\noindent where $Y_{\eta}=\theta^{\eta}(X_{\cdot+1}) :=\eta(X_1)X_1+\sum_{k=2}^{\infty}(1-\eta(X_1))\cdots(1-\eta(X_{k-1}))\eta(X_k)X_k
+\prod_{k=1}^{\infty}(1-\eta(X_k))X_{\infty}$.
\end{prop}

\begin{remark}
The term $\mathbb{E}_x[Y_{\eta}]-c({\rm Var}_x[Y_{\eta}])^{1/2}$ is interpreted as the continuation value, i.e., the value we expect to get if we choose not to stop at the current state $x$. Then (i) and (ii) tell us whether we should liquidate the whole unit or not liquidate at all by comparing the current value $x$ and the continuation value. This is in fact in the same vein as \eqref{ee7}.
\end{remark}

\begin{proof}[Proof of Proposition \ref{lemma3}]

By (\ref{thetaX}), we have \begin{align*}
\theta^{\eta}(X)=&\eta(X_0)X_0+(1-\eta(X_0))Y_{\eta},\\
\theta^{\xi\otimes\eta}(X)=&\xi(X_0)X_0+(1-\xi(X_0))Y_{\eta},
\end{align*}
and
\begin{align}
\notag K_l(x, {\xi\otimes \eta})&=x\xi(x)+(1-\xi(x))\mathbb{E}_x[Y_{\eta}]-c(1-\xi(x))({\rm Var}_x[Y_{\eta}])^{1/2}\\
\label{ee8}&=(x-(\mathbb{E}_x[Y_{\eta}]-c({\rm Var}_x[Y_{\eta}])^{1/2}))\xi(x)+\mathbb{E}_x[Y_{\eta}]-c({\rm Var}_x[Y_{\eta}])^{1/2},
\end{align}
which implies that when $\eta$ is fixed, $K_l(x, {\xi\otimes\eta}) $ is a linear function of $\xi(x)$. 

If $\eta$ is an equilibrium liquidation strategy, then according to Definition \ref{eql}, \begin{equation*}
K_l(x, \eta)=\sup_{\xi\in \cL} K_l(x, {\xi\otimes \eta}).
\end{equation*}

Therefore, if $x-(\mathbb{E}_x[Y_{\eta}]-c({\rm Var}_x[Y_{\eta}])^{1/2})>0$, then $\eta(x)=1$. If $x-(\mathbb{E}_x[Y_{\eta}]-c({\rm Var}_x[Y_{\eta}])^{1/2})<0$, then $\eta(x)=0$. $\eta(x)\in (0,1)$ only when $\eta(x)$ is a solution to the equation $x-(\mathbb{E}_x[Y_{\eta}]-c({\rm Var}_x[Y_{\eta}])^{1/2})=0$.
\end{proof}

\begin{cor}If there exists an equilibrium stopping time $\tau$ with stopping region $S$, then $\eta(x)=\mathbf{1}_{S}(x), x\in \mathbb{X}$ is an equilibrium liquidation strategy.
\end{cor}
%
%

By Proposition \ref{lemma3}, we can find an equilibrium liquidation strategy by solving a system of equations $\mathbb{E}_x[Y_{\eta}]-c({\rm Var}_x[Y_{\eta}])^{1/2}=x, \forall x\in \mathbb{X}$. The solution $\{\eta(x)\in [0,1]: x\in \mathbb{X}\}$ must be an equilibrium liquidation strategy if it exists. Other candidates of equilibrium liquidation strategies can be found by checking conditions (i) and (ii) in Proposition \ref{lemma3} when $\mathbb{E}_x[Y_{\eta}]-c({\rm Var}_x[Y_{\eta}])^{1/2}=x$ does not hold for some $x\in \mathbb{X}$. Here are some examples of mean-standard deviation problems with different sets of equilibrium liquidation strategies $\mathcal{E}.$ 

\begin{Example}\label{eg012} In this example, a unique equilibrium liquidation strategy exists, which is also an equilibrium stopping time.

Let $c=1/4$. $X$ has state space $\mathbb{X}=\{0,1, 2\}$ and the following transition matrix.\[
\begin{matrix}
 &\mathbf{0} & \mathbf{1}  & \mathbf{2} \\
\mathbf{0}& 1 & 0 & 0 \\
\mathbf{1}& 0.2 & 0.4 & 0.4\\
\mathbf{2}& 0 & 0 & 1 \\
  \end{matrix}
\]

Note that $\{0, 2\}$ are absorbing states. For any pure stopping time, its stopping region must contain the absorbing states $\{0, 2\}$. Likewise, any liquidation strategy must satisfy $\eta(0)=\eta(2)=1$.

One can easily check that the pure stopping time with stopping region $S=\{0, 2\}$ is an equilibrium stopping time with \[
\mathbb{E}_1[X_{\rho(1, S)}]-c({\rm Var}_x[X_{\rho(1, S)}])^{1/2}=\frac{4}{3}-\frac{\sqrt{2}}{6}>1.
\]

Moreover, this is the only equilibrium stopping time. If $S=\{0,1,2\}$, then \[
\mathbb{E}_1[X_{\rho(1, S)}]-c({\rm Var}_x[X_{\rho(1, S)}])^{1/2}=\frac{6}{5}-\frac{\sqrt{14}}{20}>1,
\]
which means the pure stopping time with stopping region $\{0,1,2\}$ is not an equilibrium stopping time.

Now we want to find all equilibrium liquidation strategies for this problem. The only parameter remains to be determined is $a:=\eta(1)$.

By analyzing the behavior of this Markov chain we have
\begin{align*}
&\mathbb{P}_1(X_1=0=X_n, n\ge 1)=0.2,\\
&\mathbb{P}_1(X_1=2=X_n, n\ge 1)=0.4,\\
&\mathbb{P}_1(X_1=1, X_2=0=X_n, n\ge 2)=0.4\cdot 0.2,\\
&\mathbb{P}_1(X_1=1, X_2=2=X_n, n\ge 2)=0.4\cdot 0.4,\\
&\cdots\\
&\mathbb{P}_1(X_k=1, k=1,2,\cdots, m,  X_n=0, n> m)=0.4^m\cdot 0.2,\\
&\mathbb{P}_1(X_k=1, k=1,2,\cdots, m,  X_n=2, n> m)=0.4^m\cdot 0.4,
\end{align*}
and 
\begin{align*}
&X_k=\begin{cases}1,  &k=1,2,\cdots, m,\\
0, &k\ge m+1,\\
\end{cases} \quad\quad \Rightarrow  Y_{\eta}=\eta(1)(\sum_{i=0}^{m-1}(1-\eta(1))^i),\\
&X_k=\begin{cases}1, &k=1,2,\cdots, m,\\
2, &n\ge m+1,
\end{cases} \quad\quad \Rightarrow Y_{\eta}=\eta(1)(\sum_{i=0}^{m-1}(1-\eta(1))^i)+2(1-\eta(1))^m.
\end{align*}
In conclusion, the random variable $Y_{\eta}$ has the following distribution \begin{align*}
&\mathbb{P}_1(Y_{\eta}=1-(1-\eta(1))^n)=0.4^n\cdot 0.2, \,\,\, n=0, 1, 2, \cdots,\\
&\mathbb{P}_1(Y_{\eta}=1+(1-\eta(1))^n)=0.4^n\cdot 0.4, \,\,\, n=0, 1, 2, \cdots.
\end{align*}
It can be derived that  \begin{align*}
&\mathbb{E}_1[Y_{\eta}]=\frac{0.8+0.4a}{0.6+0.4a},\\
&{\rm Var}_1[Y_{\eta}]=\frac{0.112a^2+0.256a+0.192}{(1-0.4(1-a)^2)(0.6+0.4a)^2}.
\end{align*}
Let $h(a):=\mathbb{E}_1[Y_{\eta}]-\frac{1}{4}({\rm Var}_1[Y_{\eta}])^{1/2}$. By computation we have $h(a)>1$ for all $a\in [0,1]$. So there is only one equilibrium liquidation strategy, $\eta(1)=0$, which also coincides with the unique equilibrium stopping time mentioned above.
\end{Example}
\begin{Example} \label{eg013610'} In this example, a unique equilibrium liquidation strategy exists, which is not a pure stopping time. 

Consider the example in the proof of Proposition \ref{eg013610}.
Since $\{0, 3,10\}$ are absorbing states, we have $\eta(0)=\eta(3)=\eta(10)=1$. The only parameters remain to be determined are $a:=\eta(1)$ and $b:=\eta(6)$. By the proof of Proposition \ref{eg013610}. we know that there is no equilibrium stopping time in this example. However we will see that it does have an equilibrium liquidation strategy. 

Let $g_i(a,b):=\mathbb{E}_i[Y_{\eta}]-(\mathbb{E}_i[Y_{\eta}^2]-\mathbb{E}_i[Y_{\eta}]^2)^{1/2}$ for $i=1,6$. They have explicit expressions as shown in Appendix \ref{A1}. We obtain the graph of sets $\{(a,b)\in [0,1]\times[0,1]: g_{1}(a,b)=1\}$ and $\{(a,b)\in [0,1]\times[0,1]: g_{6}(a,b)=6\}$ as following.
\begin{figure}
\centering
\includegraphics[scale=0.15]{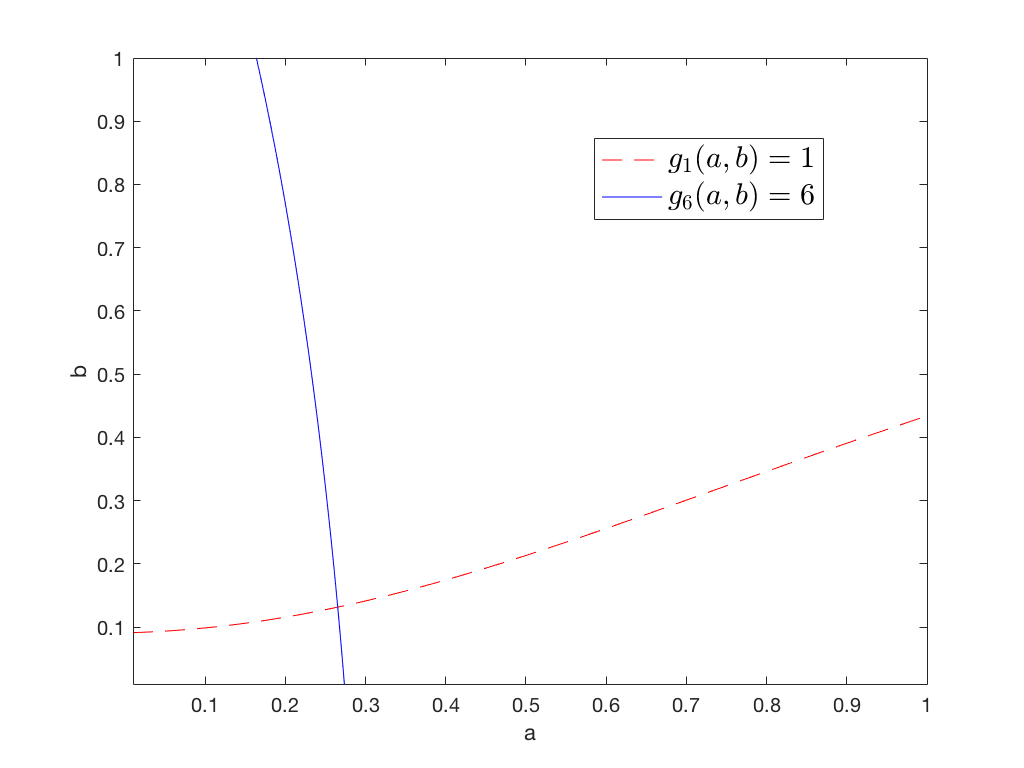}
\caption{Graph for Example~\ref{eg013610'}}
\label{fig:ffig}
\end{figure}

From Figure~\ref{fig:ffig} we observe that there exists a unique intersection of the curve $g_1(a,b)=1$  and $g_6(a,b)=6$, denoted by $(a_0, b_0)$. Then $\eta(1)=a_0, \eta(6)=b_0$ is an equilibrium liquidation strategy which is not a pure stopping time. There could be other equilibrium liquidation strategies in the following cases.

(i) If $g_1(a,0)=1$ and $g_6(a,0)>6$, then $\eta(1)=a,\eta(6)=0$ is an equilibrium liquidation strategy;

(ii) If $g_1(a,1)=1$ and $g_6(a,1)<6$, then $\eta(1)=a,\eta(6)=1$ is an equilibrium liquidation strategy;

(iii) If $g_6(0,b)=6$ and $g_1(0,b)>1$, then $\eta(1)=0,\eta(6)=b$ is an equilibrium liquidation strategy;

(iv) If $g_6(1,b)=6$ and $g_1(1,b)<1$, then $\eta(1)=1,\eta(6)=b$ is an equilibrium liquidation strategy.

However, from the above graph, we conclude that there are no solutions for $g_1(a,0)=1$, $g_1(a,1)=1$, $g_6(0,b)=6$ or $g_6(1,b)=6$. So there is only one equilibrium liquidation strategy in this example.
\end{Example}

The above two examples illustrate that an equilibrium liquidation strategy exists regardless of the existence of equilibrium stopping times. Since the equilibrium liquidation strategies are unique in Example \ref{eg012} and Example \ref{eg013610'}, they are also optimal. However, uniqueness of an equilibrium liquidation strategy and existence of an optimal equilibrium liquidation strategy are not guaranteed in general as we will show later. 

From Proposition \ref{lemma3} and \eqref{ee8}, we have that
\begin{equation}\label{ee9}
K_l(x, {\eta})=\max\{x, \,\, \E_x[Y_{\eta}]-c\V_x[Y_{\eta}]\}=
\begin{cases}
x,  &0< \eta(x)\le 1,\\
 \E_x[Y_{\eta}]-c\V_x[Y_{\eta}],  &0\leq\eta(x)<1.
\end{cases}
\end{equation}
A necessary condition for $\eta^*\in \mathcal{E}$ to be an optimal equilibrium liquidation strategy is given in the following proposition.

\begin{prop} If an equilibrium liquidation strategy $\eta^*\in \mathcal{E}$ is optimal then
\[
\mathcal{C}(\eta^*)=\bigcup_{\eta\in \mathcal{E}}\mathcal{C}(\eta),
\]
where $\mathcal{C}(\eta):=\{x\in \mathbb{X}: x<\E_x[Y_{\eta}]-c(\V_x[Y_{\eta}])^{1/2}\}$. 
\end{prop}

\begin{proof}
%
%
%
%
%
%
Suppose $\mathcal{C}(\eta^*)=\bigcup_{\eta\in \mathcal{E}}\mathcal{C}(\eta)$ does not hold. Then there exists some $\eta\in \mathcal{E}$ and some $x\in \mathbb{X}$ such that $x\in \mathcal{C}(\eta)$ and $x\notin \mathcal{C}(\eta^*)$. Then by \eqref{ee9} we have that
$$K_l(x, \eta^*)=x<\mathbb{E}_x[Y_{\eta}]-C({\rm Var}_x[Y_{\eta}])^{1/2}=K_l(x, {\eta}),$$
which contradicts that $\eta^*$ is optimal. 
\end{proof}

\begin{cor} \label{cor2} If $\bigcup_{\eta\in \mathcal{E}}\mathcal{C}(\eta)=\emptyset$, then any $\eta \in \mathcal{E}$ is an optimal equilibrium liquidation strategy.
\end{cor}

\begin{proof} If $\bigcup_{\eta\in \mathcal{E}}\mathcal{C}(\eta)=\emptyset$, then for all $\eta\in \mathcal{E}$ and all $x\in \mathbb{X}$, we have $K_l(x, {\eta})=x$. By definition, they are all optimal.
\end{proof}

The next proposition shows that the existence and uniqueness of an optimal equilibrium liquidation strategy are not guaranteed.

\begin{prop} \label{examples}The optimal equilibrium liquidation strategy may not exist. When it does exist, it may not be unique.
\end{prop}

\begin{proof} (i) We will give an example in which
there exist multiple equilibrium liquidation strategies and one of them is optimal.

Let $c=0.4$. $X$ has state space $\mathbb{X}=\{0,1, 2, 7, 9 \}$ and the following transition matrix.\[
\begin{matrix}
 &\mathbf{0} & \mathbf{1}  & \mathbf{2} & \mathbf{7}  & \mathbf{9} \\
\mathbf{0}& 1 & 0 & 0 & 0 & 0 \\
\mathbf{1}& 0.2 & 0 & 0.4 & 0.2 & 0.2 \\
\mathbf{2}& 0 & 0 & 1 & 0 & 0 \\
\mathbf{7}& 0 & 0.2 & 0 & 0 & 0.8 \\
\mathbf{9}& 0 & 0 & 0 & 0 & 1 \\
  \end{matrix}
\]

Since $\{0, 2, 9\}$ are absorbing states, we have $\eta(0)=\eta(2)=\eta(9)=1$ for all $\eta\in \mathcal{L}$.  Let $a=\eta(1)$ and $b=\eta(7)$. Let $g_i(a,b):=\mathbb{E}_i[Y_{\eta}]-(\mathbb{E}_i[Y_{\eta}^2]-\mathbb{E}_i[Y_{\eta}]^2)^{1/2}$ for $i=1, 7$. By analysis shown in Appendix \ref{A2}, we obtain the following results. 

(1) $g_1(a,b)>1$ for all $(a,b)\in [0,1]\times[0,1]$. 

(2) There exists a unique $b_0\in (0,1)$ such that $g_7(0,b_0)=7$.

(3) $g_7(0,0)>7$ and $g_7(0,1)<7$.

So there are three equilibrium liquidation strategies in total: $(a,b)=(0,0), (a,b)=(0,b_0)$ and $(a,b)=(0,1)$.  By graphs of $g_1(0,b)$ and $g_7(0,b)$ in Appendix \ref{A2}, we observe that $g_1(0,0)>g_1(0, b_0)>g_1(0,1)$ and $g_7(0,0)>g_7(0,b_0)>g_7(0,1)$, which implies that $K(x, (0,0))=\max_{\eta \in \mathcal{E}}K(x, {\eta})$ for all $x\in \mathbb{X}$. By Definition \ref{om} $(a,b)=(0,0)$ is the optimal equilibrium liquidation strategy. 

(ii) We will give an example in which there exist multiple equilibrium liquidation strategies but none of them are optimal.

Let $c=0.1$. $X$ has state space $\mathbb{X}=\{0,11, 17, 18 \}$ and the following transition matrix.\[
\begin{matrix}
 &\mathbf{0} & \mathbf{11}  & \mathbf{17} & \mathbf{18} \\
\mathbf{0}& 1 & 0 & 0 & 0 \\
\mathbf{11}& 0.1 & 0.7 & 0 & 0.2 \\
\mathbf{17}& 0 & 0.1 & 0.1 & 0.8 \\
\mathbf{18}& 0 & 0 & 0 & 1 \\
  \end{matrix}
\]

Since $\{0, 18\}$ are absorbing states, we have $\eta(0)=\eta(18)=1$ for all $\eta\in \mathcal{L}$. 
Let $a=\eta(11)$ and $b=\eta(17)$. Let $g_i(a,b):=\mathbb{E}_i[Y_{\eta}]-(\mathbb{E}_i[Y_{\eta}^2]-\mathbb{E}_i[Y_{\eta}]^2)^{1/2}$ for $i=11, 17$. By analysis shown in Appendix \ref{A3}, we obtain the following results.

(1) There is not intersection of the curve $g_{11}(a,b)=11$ and $g_{17}(a,b)=17$.

(2) There exist $0<a_1<a_2<a_3<a_4<1$ and $0<b_0<1$ such that $ g_{17}(a_1,0)=17, g_{17}(a_2, 1)=17, g_{17}(a_4,1) =17, g_{11}(a_3, 0)=g_{11}(a_3, 1)=11$, and $ g_{17}(1, b_0)=17.$

(3) $g_{11}(a_1,0)\ne 11, g_{11}(a_2, 1)\ne 11$ and $g_{11}(a_4, 1)\ne 11$.

(4) $g_{17}(a_3, 0)>17, g_{17}(a_3, 1)>17$ and $g_{11}(1,b_0)<11$.

So there are five equilibrium liquidation strategies in total. The following table summarises the values of objective functions under these equilibrium liquidation strategies.

\begin{center}
    \begin{tabular}{| l | l | l | l |}
    \hline
    $\eta: (a,b)$ & $K_l(11, \eta)$ & $K_l(17, \eta)$ \\ \hline
    $(1,1)$ & 11 & 17 \\ \hline
    $(0,1)$ & 11.1515 & 17  \\ \hline
    $(1,0)$ & 11 & 17.0022  \\ \hline
    $(a_3,0)$ & 11 & $\approx$ 17.0212 \\ \hline
    $(1, b_0)$ & 11 & 17\\ 
    \hline
    \end{tabular}
\end{center}

The table shows that there is no optimal equilibrium liquidation strategy.

(iii) We will give an example in which there are two equilibrium liquidation strategies and both are optimal.

Let $c=0.5$. $X$ has state space $\mathbb{X}=\{0,1,4\}$ and the following transition matrix. \[
\begin{matrix}
 &\mathbf{0} & \mathbf{1}  & \mathbf{4} \\
\mathbf{0}& 1 & 0 & 0 \\
\mathbf{1}& 0.1 & 0.8 & 0.1\\
\mathbf{4}& 0 & 0 & 1  \\
  \end{matrix}
\]

Since $\{0, 4\}$ are absorbing states, we have $\eta(0)=\eta(4)=1$ for all $\eta$. Let $a=\eta(1)$ and we  obtain function $h(a)=\E_1[Y_{\eta}]-c\V_x[Y_{\eta}]$ as shown in Appendix \ref{A4}. $h(a)$ is decreasing on the interval $[0,1]$ and $h(0)=1, h(1)=0.7101<1$. By definition, there are two equilibrium liquidation strategies in total $\eta(1)=0$ and $\eta(1)=1$. By Corollary \ref{cor2}, $K_l(1,{\eta} )=1$ for both equilibrium liquidation strategies, so they are both optimal.
\end{proof}

Since the existence of  optimal equilibrium liquidation strategy is not guaranteed, we naturally turn to the concept of Pareto optimality.

\begin{defn} $\eta^*\in \mathcal{E}$ is called a Pareto optimal equilibrium liquidation strategy if there is no $\eta\in\mathcal{E}$ such that 

(i) $\forall\,  x\in \mathbb{X}, K_l(x, \eta)\ge K_l(x, \eta^*)$;

(ii) $\exists\,  x\in \mathbb{X}, K_l(x, \eta)> K_l(x, \eta^*)$
\end{defn}

\begin{remark} In the second example in proof of Proposition \ref{examples}, the optimal equilibrium  liquidation strategy does not exist, but $(0,1)$ and $(a_3,0)$ are both Pareto optimal equilibrium liquidation strategies.
\end{remark}  

\begin{prop} A Pareto optimal equilibrium liquidation strategy always exists.
\end{prop}
\begin{proof} Consider the following optimization problem\begin{equation}
\sup_{\eta\in \mathcal{E}} \sum_{x\in \mathbb{X}} K_l(x, \eta). \label{sup}
\end{equation}

Then any maximizer $\eta^*$ of this problem is a Pareto optimal equilibrium liquidation strategy. Otherwise, there exists $\eta'\in \mathcal{E}$ such that \[
K_l(x, \eta')\ge K_l(x, \eta^*), \forall x\in \mathbb{X};\,\, K_l(x_0, \eta')> K_l(x_0, \eta^*), \exists x_0\in \mathbb{X}; 
\] 
then $ \sum_{x\in \mathbb{X}} K_l(x, \eta')> \sum_{x\in \mathbb{X}} K_l(x, \eta^*),$ which contradicts that $\eta^*$ is the maximizer of the problem. 

Next we will show such maximizer exists and is in $\mathcal{E}$. Let $\eta^n$ be the $\frac{1}{n}$-optimizer of (\ref{sup}). Then there exists $\eta^*\in \mathcal{L}$ such that up to a subsequence $\eta^n\to \eta^*$.
Since $\{\eta^n\}_{n\in \mathbb{N}}$ are all equilibrium liquidation strategies, $K_l(x, \xi\otimes \eta^n)\le K_l(x, \eta^n)$ holds for all $x\in \mathbb{X}$ and all $\xi\in \cL$. By the continuity of the mappings $\eta\to K_l(x, \xi\otimes \eta)$ and  $\eta\to K_l(x, \eta)$,   $K_l(x, \xi\otimes \eta^*)\le K_l(x, \eta)$ also holds for all $x\in \mathbb{X}$ and all $\xi\in \cL$, i.e., $\eta^*$ is also an equilibrium liquidation strategy. Again by the continuity of mapping $\eta\to K_l(x, {\eta})$, $\eta^*$ is the maximizer of (\ref{sup}).
\end{proof}

\section{Mean-Variance Problem}\label{sec_MV}

As what we have done in mean-standard deviation problem, we will analyze different types of subgame perfect Nash equilibrium in mean-variance problems. More specifically, we define equilibrium stopping time, equilibrium randomized stopping strategy and equilibrium liquidation strategy for mean-variance problems as the following. 

\begin{defn} A pure Markov stopping time $\tau$ with stopping region $S$ is said to be an equilibrium stopping time for the mean-variance problem if
\begin{equation*}
x\ge J_p(x, \rho(x, S)),\,\, \forall x\in S \quad \quad {\rm and}\quad \quad x\leq J_p(x, \rho(x, S)),\,\, \forall x\notin S, 
\end{equation*}
where $\rho(x, S):=\inf\{n\ge 1: X^x_n\in S\}$ and $J_p(x, \rho(x, S)):=\E_x[X_{\rho(x, S)}]-c\V_x[X_{\rho(x, S)}]$. 
\end{defn}

\begin{defn}  A randomized stopping time $\mathbf{p}\in \mathcal{P}$ is said to be an equilibrium randomized stopping time for  the mean-variance problem, if for any mapping $\mathbf{q}: \mathbb{X}\to [0,1]$,\[
J_r(x, \mathbf{q}\otimes \mathbf{p})\le J_r(x,\p\otimes\mathbf{p}), \,\,\, \forall x\in \mathbb{X},
\]
where $J_r(x,\q\otimes\mathbf{p}):=\E_x[X_{\q\otimes\mathbf{p}}]-c\V_x[X_{\q\otimes\mathbf{p}}]$. 
\end{defn}

\begin{defn} \label{def3}A liquidation strategy $\theta=\theta^{\eta}\in \mathcal{L}$ is said to be an equilibrium liquidation strategy for the mean-variance problem
if for any mapping $\xi: \mathbb{X}\to [0,1]$, we have\begin{equation*}
J_l(x, {\xi\otimes\eta})\le J_l(x, \eta\otimes{\eta}),\,\,\, \forall x\in \mathbb{X},
\end{equation*}
where $J_l(x, \xi\otimes\eta):=\E_x[\theta^{\xi\otimes\eta}(X)]-c\V_x[\theta^{\xi\otimes\eta}(X)]$. 
\end{defn}

\begin{prop} An equilibrium stopping time for mean-variance problem may not exist.
\end{prop}

\begin{proof}
We will prove this by giving a counterexample.
Let $c=\frac{21}{50}$ and $X$ has the following transition matrix.
$$
\begin{array}{ccccc}
&{\bf0}&{\bf1}&{\bf2}&{\bf3}\\
{\bf0}&1&0&0&0\\
{\bf1}&\frac{1}{3}&0&\frac{1}{3}&\frac{1}{3}\\
{\bf2}&0&\frac{1}{3}&0&\frac{2}{3}\\
{\bf3}&0&0&0&1
\end{array}
$$
Suppose there exists an equilibrium stopping time with stopping region $S\subset\{0,1,2,3\}$. Denote $H(\cdot,S)=J_p(\cdot,\rho(\cdot,S))$. We consider the following four cases.

\textbf{Case 1: $1,2\in S$.} We have
$$\PP_1(X_{\rho(1,S)}=0)=\PP_1(X_{\rho(1,S)}=2)=\PP_1(X_{\rho(1,S)}=3)=\frac{1}{3},$$
and
$$\E_1\left[X_{\rho(1,S)}\right]=\frac{5}{3},\quad \E_1\left[X_{\rho(1,S)}^2\right]=\frac{13}{3},\quad\text{and thus}\quad H(1,S)=\frac{76}{75}>1,$$
that yields a contradiction.

\textbf{Case 2: $1\notin S$ and $2\in S$.} We have
$$\PP_2(X_{\rho(2,S)}=0)=\PP_2(X_{\rho(2,S)}=2)=\frac{1}{9},\quad\PP_2(X_{\rho(2,S)}=3)=\frac{7}{9},$$
and
$$\E_2\left[X_{\rho(2,S)}\right]=\frac{23}{9},\quad \E_2\left[X_{\rho(2,S)}^2\right]=\frac{67}{9},\quad\text{and thus}\quad H(2,S)=\frac{1466}{675}>2,$$
that yields a contradiction.

\textbf{Case 3: $1\in S$ and $2\notin S$.} We have that
$$\PP_2(X_{\rho(2,S)}=1)=\frac{1}{3},\quad\PP_2(X_{\rho(2,S)}=3)=\frac{2}{3},$$
and
$$\E_2\left[X_{\rho(2,S)}\right]=\frac{7}{3},\quad \E_2\left[X_{\rho(2,S)}^2\right]=\frac{19}{3},\quad\text{and thus}\quad H(2,S)=\frac{147}{75}<2,$$
that yields a contradiction.

\textbf{Case 4: $1,2\notin S$.} We have that
$$\PP_1(X_{\rho(1,S)}=0)=\frac{3}{8},\quad\PP_1(X_{\rho(1,S)}=3)=\frac{5}{8},$$
and
$$\E_1\left[X_{\rho(1,S)}\right]=\frac{15}{8},\quad \E_1\left[X_{\rho(1,S)}^2\right]=\frac{45}{8},\quad\text{and thus}\quad H(1,S)=\frac{633}{640}<1,$$
that yields a contradiction.
\end{proof}

Besides, using a proof similar to the case of mean-standard deviation problem, we can show that an equilibrium randomized stopping time for mean-variance problem exists if and only if it is an equilibrium stopping time.

Now we will focus on the equilibrium liquidation strategy for mean-variance problem as defined in Definition \ref{def3}, although it is not a proper definition as pointed out in Remark \ref{rmk3}. Following an argument similar to the one in the proof of Theorem \ref{thm1}, we can prove that there exists an equilibrium liquidation strategy in mean-variance problem. 
The next proposition shows that in contrast to mean-standard deviation problem, the equilibrium liquidation strategy is not a generalization of equilibrium stopping time in mean-variance problem, although equilibrium liquidation strategies can be thought of as a relaxation of equilibrium stopping times.
\begin{prop} 
An equilibrium stopping time may not be an equilibrium liquidation strategy.
\end{prop}

\begin{proof}
Consider the Markov process in Example \ref{eg012} and let $c=0.25$. It can be shown that there is only one equilibrium stopping time with stopping region $S=\{0,2\}$ and $\E_x[X_S]-c\V_x[X_S]=\frac{10}{9}$. Next we will show that the corresponding liquidation strategy $\eta(1)=0$ is not an equilibrium liquidation strategy. Since $\E_1[Y_{\eta}]=\frac{4}{3}$ and $\V_1[Y_{\eta}]=\frac{8}{9}$, we have \[
J_l(1, {\xi\otimes\eta})=\xi(1)+\frac{4}{3}(1-\xi(1))-\frac{2}{9}(1-\xi(1))^2.
\]

It is easy to check that $\max_{\xi}J_l(1, {\xi\otimes\eta})=\frac{9}{8}>\frac{10}{9}$ and the maximum is attained at $\xi(1)=\frac{1}{4}$ instead of $0$. By definition $\eta$ is not an equilibrium liquidation strategy.
\end{proof}

This result illustrates that the equilibrium liquidation strategy for mean-variance problem is not a proper definition as we briefly discuss next.

\begin{remark} \label{rmk3} Definition \ref{def3} seems to be reasonable as an analogy of Definition \ref{eql}. However, since there is no scaling effect in the mean-variance problem, this definition has deviated from the concept of subgame perfect Nash equilibrium. 
For example, at time $t=0$, the objective function is \[
\E_{X_0}[\theta(X)]-c\V_{X_0}[\theta(X)].
\]
If we liquidate $1-\alpha$ proportion of the asset, then the objective function at time $t=1$ would become\[
\E_{X_1}[\alpha\theta(X)]-c\V_{X_1}[\alpha\theta(X)]=\alpha(\E_{X_1}[\theta(X)]-c\alpha \V_{X_1}[\theta(X)]),
\]
i.e., the preference of player ``$X_1$'' becomes  $\E_{X_1}[\theta(X)]-c\alpha \V_{X_1}[\theta(X)]$ instead of $\E_{X_1}[\theta(X)]-c \V_{X_1}[\theta(X)]$. Generally, the proportion of asset remaining, $\alpha$, is decreasing as time goes on, therefore we are faced with different problems with different parameter $c$  at different time, even if the initial state remains the same. Definitions of equilibrium liquidation strategies in Definition \ref{def3} and Definition \ref{eql} only make sense when the objective function remains the same for the same initial state $x$.

A possible improved definition is incorporating the remaining component, i.e., to enlarge the strategy set such that it depends on the state as well as the remaining component of the asset. However, this expansion makes the set of players an uncountable set (instead of identifying the players with the states of the Markov chain, we will need to use an additional variable which is not discrete). This is an intergenerational problem with exhaustible resources. Such a problem is beyond the scope of this paper and will be left for future research. But we should emphasize that one of the main messages of our paper is that mean-standard deviation problem is more appropriate and for this criterion such an extension of the state space is not necessary. 

\end{remark}


\section{Comparison with Static Optimal Stopping Time}\label{sec_compare}

In this section, we want to compare the pre-commitment strategy with the equilibrium liquidation strategy. It is obvious that  given any current state $x$, the static optimality is no less than the value of $K_p(x, \tau)$ where $\tau$ is an equilibrium stopping time. However this may not be the case when we compare static optimal stopping times with equilibrium liquidation strategies. As we have discussed in Example \ref{eg0123}, a liquidation strategy may produce larger value than the static optimal stopping time does. However in Example \ref{eg0123}, the liquidation strategy is not an equilibrium. The following examples show that an equilibrium liquidation strategy may produce larger value than the static optimal stopping time does in both mean-standard deviation problem and mean-variance problem. The intuitive reason is that a liquidation strategy allows for selling parts of an asset over time, while the static optimal stopping time problem relies on the assumption that the whole asset must be sold at exactly one point in time. 
 
\begin{Example} \label{eg4.1}
Consider the first example in the proof of Proposition \ref{examples}. 

For mean-standard deviation problem,
\[\sup_{\tau\in \mathcal{T}}K_p(1, \tau)=K_p(1, \tau')=2.6940, \,\,\,\, \sup_{\tau\in \mathcal{T}}K_p(7, \tau)=K_p(7, \tau')=7.0187.\]
where $\tau'=\inf\{n\ge 0: X^1_n\in\{0,2,9\}\}$.

The optimal  equilibrium liquidation strategy $\eta$ is the same as $\tau'$. We have \[
K_l(1,{\eta})=\sup_{\tau\in \mathcal{T}}K_p(1, \tau), \,\,\,\,  K_l(7, {\eta})=\sup_{\tau\in \mathcal{T}}K_p(7, \tau).
\]

For mean-variance problem,
\[\sup_{\tau}J_p(1, \tau)=J_p(1, \tau'')=1,\,\, \sup_{\tau}J_p(7, \tau)=J_p(7, \tau'')=7. \]
where $\tau''=0$.

The unique equilibrium liquidation strategy is $\eta(1)=a'\in (0,1), \eta(7)=b'\in (0,1)$ where $a'\approx 0.6778$ and $b'\approx 0.9089$. Details on finding the equilibrium liquidation strategy can be found in Appendix \ref{A5}.  We have \[
J_l(1,{\eta})=2.6438>\sup_{\tau\in \mathcal{T}}J_p(1, \tau), \,\,\,\, J_l(7, {\eta})=6.4521<\sup_{\tau\in \mathcal{T}}J_p(7, \tau).
\]
\end{Example}

\begin{Example} \label{eg4.2}
Consider the second example in the proof of Proposition \ref{examples}. 

For mean-standard deviation problem,\[\sup_{\tau\in \mathcal{T}}K_p(11, \tau)=K_p(11, \tau')=K_p(11, \tau'')=11.1515, \]
where $\tau'=\inf\{n\ge 0: X^1_n\in\{0,18\}\}, \tau''=\inf\{n\ge 0: X^1_n\in\{0,17,18\}\}$, and
\[\sup_{\tau\in \mathcal{T}}K_p(17, \tau)=K_p(17, \tau''')=17.0022, \]
where $\tau'''=\inf\{n\ge 0: X^1_n\in\{0,11,18\}\}$.

There is an equilibrium liquidation strategy $\eta(11)=a_3\in (0,1), \eta(17)=0$. We have \[
K_l(11, {\eta})=11<\sup_{\tau\in \mathcal{T}}K_p(11, \tau), \,\,\,\, K_l(17, {\eta})=17.0212>\sup_{\tau\in \mathcal{T}}K_p(17, \tau).
\]

For mean-variance problem,
\[\sup_{\tau\in \mathcal{T}}J_p(11, \tau)=J_p(11, \tau')=11, \,\,\,\, \sup_{\tau\in \mathcal{T}}J_p(17, \tau)=J_p(17, \tau')=17, \]
where $\tau'=0$.

The unique  equilibrium liquidation strategy is $\eta(11)=a'\in (0,1), \eta(17)=b'\in (0,1)$ where $a'\approx 0.9312$ and $b'\approx 0.7629$. Details on finding the equilibrium liquidation strategy can be found in Appendix \ref{A5}. We have \[
J_l(11, {\eta})=10.8365<\sup_{\tau\in \mathcal{T}}J_p(11, \tau), \,\,\,\, J_l(17, {\eta})=16.9981<\sup_{\tau\in \mathcal{T}}J_p(17, \tau).
\]
\end{Example}

\appendix
\section{Computation details}  \label{A}
\subsection{Example \ref{eg013610'}} \label{A1}

We first analyze all the possible trajectories of Markov chain $X$ when starting from $1$ and $6$.

\textbf{Case 1:}
$X:1\to 6\to 1\to 6 \to 1\to\cdots \to 6 \to 1\to 0$.

$\Rightarrow Y_{\eta}=0$ for $k=0$ and $Y_{\eta}=a\sum_{i=1}^{k}(1-a)^{i-1}(1-b)^{i}+6b\sum_{i=0}^{k-1}(1-a)^i(1-b)^i$ for $k\ge 1$. Then $Y_{\eta}=c(1-(1-a)^k(1-b)^k)$ with probability $0.2^{2k+1}$  for $k\ge 0$ where $c=\frac{a+6b-ab}{1-(1-a)(1-b)}$.

\textbf{Case 2:}
$X:1\to 6\to 1\to 6 \to 1\to\cdots \to 6 \to 1\to 3$.

$\Rightarrow Y_{\eta}=3$ for $k=0$ and $Y_{\eta}=a\sum_{i=1}^{k}(1-a)^{i-1}(1-b)^{i}+6b\sum_{i=0}^{k-1}(1-a)^i(1-b)^i+3(1-a)^k(1-b)^k$ for $k\ge 1$. Then $Y_{\eta}=c(1-(1-a)^k(1-b)^k)+3(1-a)^k(1-b)^k$  with probability $2\times0.2^{2k+1}$ for $k\ge 0$.

\textbf{Case 3:}
$X:1\to 6\to 1\to 6 \to 1\to\cdots \to 6 \to 1\to 10$.

$\Rightarrow Y_{\eta}=10$ for $k=0$ and $Y_{\eta}=a\sum_{i=1}^{k}(1-a)^{i-1}(1-b)^{i}+6b\sum_{i=0}^{k-1}(1-a)^i(1-b)^i+10(1-a)^k(1-b)^k$ for $k\ge 1$. Then $Y_{\eta}=c(1-(1-a)^k(1-b)^k)+10(1-a)^k(1-b)^k$  with probability $0.2^{2k+1}$ for $k\ge 0$.

\textbf{Case 4:}
$X:1\to 6\to 1\to 6 \to 1\to\cdots \to 6 \to 10$.

$\Rightarrow  Y_{\eta}=a\sum_{i=1}^{k}(1-a)^{i-1}(1-b)^{i}+6b\sum_{i=0}^{k}(1-a)^i(1-b)^i+10(1-a)^k(1-b)^{k+1}$ for $k\ge 0$. Then $Y_{\eta}=c(1-(1-a)^k(1-b)^k)+(10-4b)(1-a)^k(1-b)^k$  with probability $0.8\times0.2^{2k+1}$ for $k\ge 0$.

\textbf{Case 5:}
$X:6\to 1\to 6 \to 1\to\cdots \to 6 \to 1\to 0$ for $k\ge 0$.

$\Rightarrow Y_{\eta}=a$ for $k=0$ and $Y_{\eta}=a\sum_{i=0}^{k}(1-a)^{i}(1-b)^{i}+6b\sum_{i=0}^{k-1}(1-a)^{i+1}(1-b)^i$ for $k\ge 1$. Then $Y_{\eta}=d(1-(1-a)^k(1-b)^k)+a (1-a)^k(1-b)^k$ with probability $0.2^{2k+2}$ for $k\ge 0$ where $d=\frac{a+6b-6ab}{1-(1-a)(1-b)}$.

\textbf{Case 6:}
$X:6\to 1\to 6 \to 1\to\cdots \to 6 \to 1\to 3$.

$\Rightarrow Y_{\eta}=a+3(1-a)$ for $k=0$ and $Y_{\eta}=a\sum_{i=0}^{k}(1-a)^{i}(1-b)^{i}+6b\sum_{i=0}^{k-1}(1-a)^{i+1}(1-b)^i+3(1-a)^{k+1}(1-b)^k$ for $k\ge 1$. Then $Y_{\eta}=d(1-(1-a)^k(1-b)^k)+(3-2a)(1-a)^k(1-b)^k$  with probability $2\times0.2^{2k+2}$ for $k\ge 0$.

\textbf{Case 7:}
$X:6\to 1\to 6 \to 1\to\cdots \to 6 \to 1\to 10$.

$\Rightarrow Y_{\eta}=a+10(1-a)$ for $k=0$ and $Y_{\eta}=a\sum_{i=0}^{k}(1-a)^{i}(1-b)^{i}+6b\sum_{i=0}^{k-1}(1-a)^{i+1}(1-b)^i+10(1-a)^{k+1}(1-b)^k$ for $k\ge 1$. Then $Y_{\eta}=d(1-(1-a)^k(1-b)^k)+(10-9a)(1-a)^k(1-b)^k$  with probability $0.2^{2k+2}$ for $k\ge 0$.

\textbf{Case 8:}
$X:6\to 1\to 6 \to 1\to\cdots \to 6 \to 10$.

$\Rightarrow Y_{\eta}=10$ for $k=0$ and $Y_{\eta}=a\sum_{i=0}^{k-1}(1-a)^{i}(1-b)^{i}+6b\sum_{i=0}^{k-1}(1-a)^{i+1}(1-b)^i+10(1-a)^{k}(1-b)^k$ for $k\ge 1$. Then $Y_{\eta}=d(1-(1-a)^k(1-b)^k)+10(1-a)^k(1-b)^k$  with probability $0.8\times0.2^{2k}$ for $k\ge 0$.

From the above, we can conclude that 

(1) When $X_0=1$, \begin{align*}
& \mathbb{P}_1(Y_{\eta}=c-ct^k)=0.2^{2k+1}, k\ge 0,\\
& \mathbb{P}_1(Y_{\eta}=c+(3-c)t^k)=2\times0.2^{2k+1}, k\ge 0,\\
& \mathbb{P}_1(Y_{\eta}=c+(10-c)t^k)=0.2^{2k+1}, k\ge 0,\\
& \mathbb{P}_1(Y_{\eta}=c+(10-4b-c)t^k)=0.8\times0.2^{2k+1}, k\ge 0,
\end{align*}
where $t=(1-a)(1-b)$, and \begin{align*}
&\mathbb{E}_1[Y_{\eta}]=c+\frac{4.8-0.96c-0.64b}{1-0.04t},\\
&\mathbb{E}_1[Y_{\eta}^2]=c^2+\frac{-1.92c^2+9.6c-1.28cb}{1-0.04t}+\frac{0.96c^2-9.6c+39.6-12.8b+1.28bc+2.56b^2}{1-0.04t^2}.
\end{align*}

(2) When $X_0=6$, \begin{align*}
& \mathbb{P}_6(Y_{\eta}=d+(a-d)t^k)=0.2^{2k+2}, k\ge 0,\\
& \mathbb{P}_6(Y_{\eta}=d+(3-2a-d)t^k)=2\times0.2^{2k+2}, k\ge 0,\\
& \mathbb{P}_6(Y_{\eta}=d+(10-9a-d)t^k)=0.2^{2k+2}, k\ge 0,\\
& \mathbb{P}_6(Y_{\eta}=d+(10-d)t^k)=0.8\times0.2^{2k}, k\ge 0,
\end{align*}
where $t=(1-a)(1-b)$, and\begin{align*}
&\mathbb{E}_6[Y_{\eta}]=d+\frac{8.64-0.96d-0.48a}{1-0.04t},\\
&\mathbb{E}_6[Y_{\eta}^2]=d^2+\frac{-1.92d^2+17.28d-0.96ad}{1-0.04t}+\frac{0.96d^2-17.28d+84.72+3.6a^2+0.96ad-8.16a}{1-0.04t^2}.
\end{align*}

Then we will obtain the result in Example \ref{eg013610'}.

\subsection{The first example in Proposition \ref{examples} } \label{A2}
Since the transition matrix in this example is the same as Example \ref{eg013610'}, by following a similar analysis of  $X$'s trajectories, we have

(1) When $X_0=1$,
\begin{align*}
& \mathbb{P}_1(Y_{\eta}=c-ct^k)=0.2^{2k+1}, k\ge 0,\\
& \mathbb{P}_1(Y_{\eta}=c+(2-c)t^k)=2\cdot 0.2^{2k+1}, k\ge 0,\\
& \mathbb{P}_1(Y_{\eta}=c+(9-c)t^k)=0.2^{2k+1}, k\ge 0,\\
& \mathbb{P}_1(Y_{\eta}=c+(9-2b-c)t^k)=0.8\cdot 0.2^{2k+1}, k\ge 0,
\end{align*}
where $c=\frac{a+7b-ab}{1-(1-a)(1-b)}$ and $t=(1-a)(1-b)$. Then we have\begin{align*}
&\mathbb{E}_1[Y_{\eta}]=c+\frac{0.2(20.2-4.8c-1.6b)}{1-0.04t},\\
&\mathbb{E}_1[Y_{\eta}^2]=c^2+\frac{0.4(-4.8c^2+20.2c-1.6cb)}{1-0.04t}+\frac{0.2(c^2+2(2-c)^2+(9-c)^2+0.8(9-2b-c)^2)}{1-0.04t^2}.
\end{align*}

(2) When $X_0=7$,
\begin{align*}
& \mathbb{P}_7(Y_{\eta}=d+(a-d)t^k)=0.2^{2k+2}, k\ge 0,\\
& \mathbb{P}_7(Y_{\eta}=d+(2-a-d)t^k)=2\times0.2^{2k+2}, k\ge 0,\\
& \mathbb{P}_7(Y_{\eta}=d+(9-8a-d)t^k)=0.2^{2k+2}, k\ge 0,\\
& \mathbb{P}_7(Y_{\eta}=d+(9-d)t^k)=0.8\times0.2^{2k}, k\ge 0,
\end{align*}
where $d=\frac{a+7b-7ab}{1-(1-a)(1-b)}$ and $t=(1-a)(1-b)$. Then we have\begin{align*}\
&\mathbb{E}_7[Y_{\eta}]=d+\frac{0.04(193-24d-9a)}{1-0.04t},\\
&\mathbb{E}_7[Y_{\eta}^2]=d^2+\frac{0.08(-24d^2+193d-9ad)}{1-0.04t}+\\
&\quad \quad \quad \quad \frac{0.04((a-d)^2+2(2-a-d)^2+(9-8a-d)^2+20(9-d)^2)}{1-0.04t^2}.
\end{align*}

Furthermore, we can find that for any $a,b\in [0,1]$, $g_1(a,b)>1$, which implies that $a=0$. By plotting the graphs of $g_1(0,b)$ and $g_7(0,b)$ as functions of $b\in [0,1]$, we obtain

\begin{center}
\includegraphics[scale=0.3]{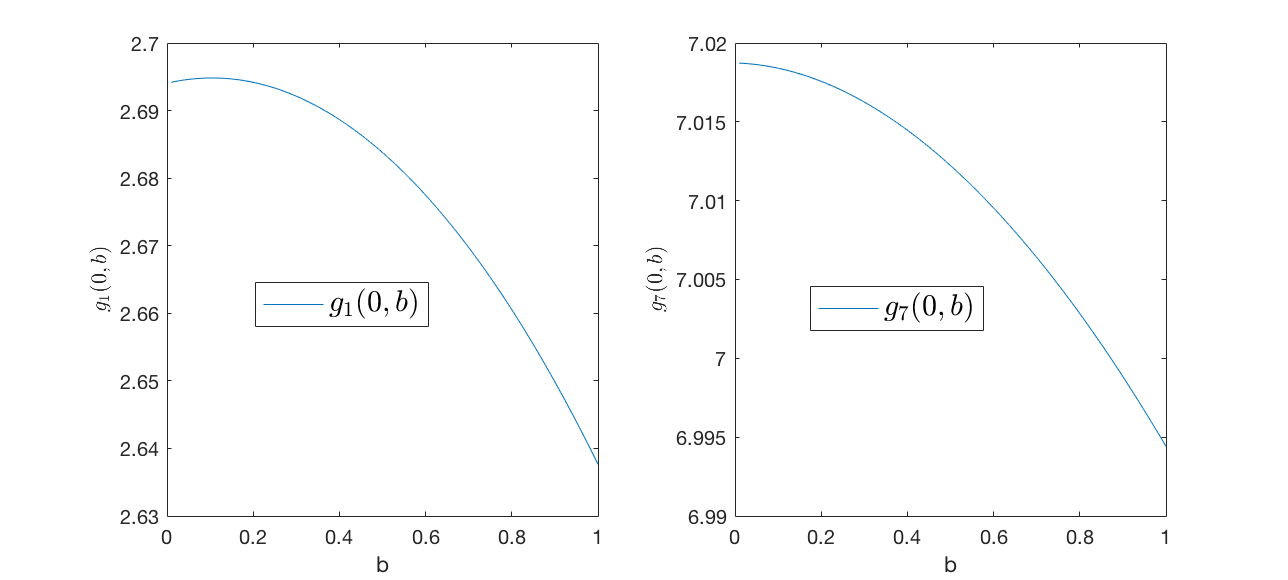}
\end{center}
which follows the discussion in part (i) of the proof of Proposition \ref{examples}.

\subsection{The second example in Proposition \ref{examples}} \label{A3}

It can be computed directly that there are three equilibrium stopping times. They can be written in the form of liquidation strategies: $(a,b)=(1,1), (a,b)=(0,1)$, and $(a,b)=(1,0)$. To check they are  indeed equilibria, compute \begin{align*}
&g_{11}(1,1)=10.8330<11, \,\, g_{17}(1,1)=16.9912<17,\\
&g_{11}(0,1)=11.1515>11, \,\, g_{17}(0,1)=16.9774<17,\\
&g_{11}(1,0)=10.8330<11, \,\, g_{17}(1,0)=17.0022>17,
\end{align*}
where $g_{11}(a,b)=\mathbb{E}_{11}[Y_{\eta}]-0.1(\mathbb{E}_{11}[Y_{\eta}^2]-\mathbb{E}_{11}[Y_{\eta}]^2)^{1/2}$ and $g_{17}(a,b)=\mathbb{E}_{17}[Y_{\eta}]-0.1(\mathbb{E}_{17}[Y_{\eta}^2]-\mathbb{E}_{17}[Y_{\eta}]^2)^{1/2}$. 

Also notice that $(a,b)=(0,0)$ is not a equilibrium liquidation strategy since $g_{17}(0,0)=16.9934<17$. 

To find all equilibrium liquidation strategies, we need to analyze all the possible trajectories of this Markov chain when starting from $11$ and $17$.

\textbf{Case 1:}
$X: 11 \to 11 \to \cdots \to 11 \to 0$. 

Then $Y_\eta=0$ with probability 0.1 and $Y_{\eta}=11a\sum_{i=0}^{k-1}(1-a)^i$ with probability $0.1\cdot 0.7^k$ for $k\ge 1$. 

\textbf{Case 2:}
$X: 11 \to 11 \to \cdots \to 11 \to 18$. 

Then $Y_\eta=18$ with probability 0.2 and $Y_{\eta}=11a\sum_{i=0}^{k-1}(1-a)^i+18(1-a)^k$ with probability $0.2\cdot 0.7^k$ for $k\ge 1$. 

\textbf{Case 3:}
$X: 17 \to 17 \to \cdots \to 17 \to 18$. 

Then $Y_\eta=18$ with probability 0.8 and $Y_{\eta}=17b\sum_{i=0}^{k-1}(1-b)^i+18(1-b)^k$ with probability $0.8\cdot 0.1^k$ for $k\ge 1$. 

\textbf{Case 4:}
$X: 17 \to 17 \to \cdots \to 17 \to 11 \to 11 \to \cdots \to 11 \to 0$. 

Then $Y_\eta=11a\sum_{j=0}^{m}(1-a)^j$ with probability $0.01\cdot 0.7^m$ for $m\ge 0$ and $Y_{\eta}=17b\sum_{i=0}^{k-1}(1-b)^i+(1-b)^k11a\sum_{j=0}^{m}(1-a)^j$ with probability $0.01\cdot 0.7^m\cdot 0.1^k$ for $k\ge 1, m\ge 0$. 

\textbf{Case 5:}
$X: 17 \to 17 \to \cdots \to 17 \to 11 \to 11 \to \cdots \to 11 \to 18$. 

Then $Y_\eta=11a\sum_{j=0}^{m}(1-a)^j+18(1-a)^{m+1}$ with probability $0.02\cdot 0.7^m$ for $m\ge 0$ and $Y_{\eta}=17b\sum_{i=0}^{k-1}(1-b)^i+(1-b)^k(11a\sum_{j=0}^{m}(1-a)^j+18(1-a)^{m+1})$ with probability $0.02\cdot 0.7^m\cdot 0.1^k$ for $k\ge 1, m\ge 0$. 

From the above, we can conclude that 

(1) When $X_0=11$,
\begin{align*}
& \mathbb{P}_{11}(Y_{\eta}=11-11(1-a)^k)=0.1\cdot 0.7^k, \quad k\ge 0,\\
& \mathbb{P}_{11}(Y_{\eta}=11+7(1-a)^k)=0.2\cdot 0.7^{k}, \quad k\ge 0,
\end{align*} 
and \begin{align*}
&\mathbb{E}_{11}[Y_{\eta}]=11+\frac{0.3}{0.3+0.7a},\\
&\mathbb{E}_{11}[Y_{\eta}^2]=11^2+\frac{6.6}{0.3+0.7a}+\frac{21.9}{1-0.7(1-a)^2}.
\end{align*}

(2) When $X_0=17$,
\begin{align*}
& \mathbb{P}_{17}(Y_{\eta}=17+(1-b)^k)=0.8\cdot 0.1^{k},\quad  k\ge 0\\
& \mathbb{P}_{17}(Y_{\eta}=17-(6+11(1-a)^{m+1})(1-b)^k)=0.01\cdot 0.7^m\cdot 0.1^k,\quad  k\ge 0, m\ge 0,\\
& \mathbb{P}_{17}(Y_{\eta}=17-(6-7(1-a)^{m+1})(1-b)^k)=0.02\cdot 0.7^m\cdot 0.1^k,\quad  k\ge 0, m\ge 0,
\end{align*}
and \begin{align*}
&\mathbb{E}_{17}[Y_{\eta}]=17+(0.2+\frac{0.03(1-a)}{0.3+0.7a})\frac{1}{0.9+0.1b},\\
&\mathbb{E}_{17}[Y_{\eta}^2]=17^2+(0.2+\frac{0.03(1-a)}{0.3+0.7a})\frac{34}{0.9+0.1b}+\\
&\quad \quad \quad \quad  (4.4-\frac{0.36(1-a)}{0.3+0.7a}+\frac{2.19(1-a)^2}{1-0.7(1-a)^2})\frac{1}{1-0.1(1-b)^2}.
\end{align*}

The sets $\{(a,b)\in [0,1]\times[0,1] : g_{11}(a,b)=11 \}$ and  $\{(a,b)\in [0,1]\times[0,1] : g_{17}(a,b)=17 \}$ are shown as the following. 

\begin{center}
\includegraphics[scale=0.25]{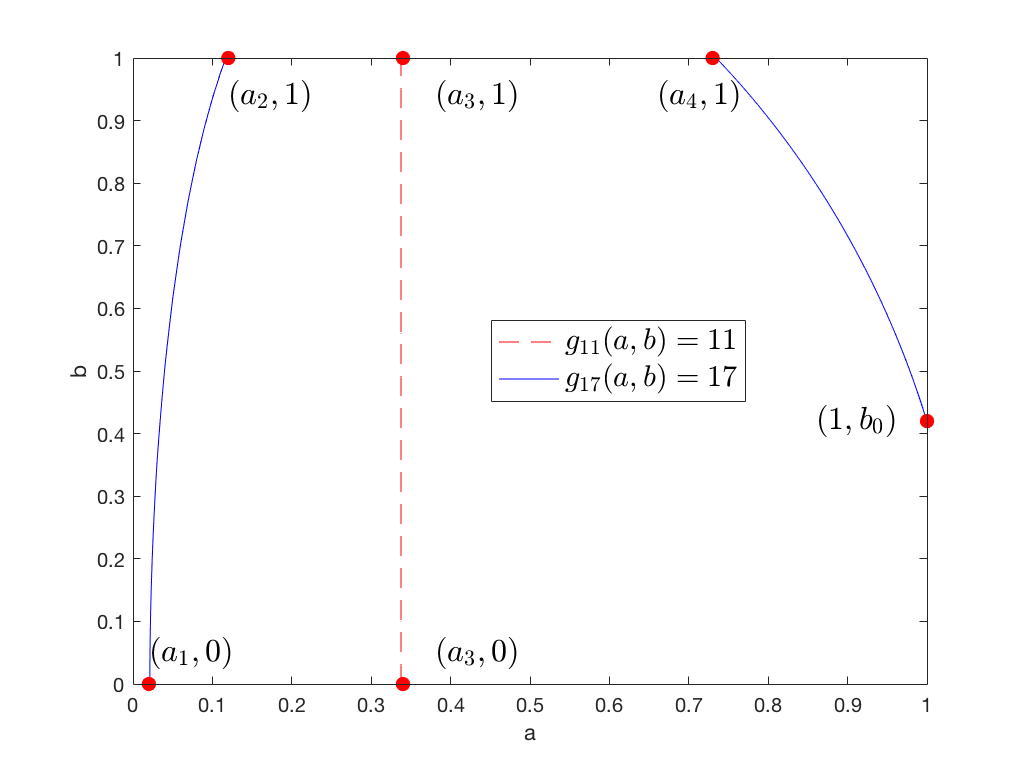}
\end{center}

This graph shows that the curves $g_{11}(a,b)=11$ and $g_{17}(a,b)=17$ do not intersect. So the candidates for equilibrium liquidation strategies only lie on the boundary of $[0,1]\times[0,1]$. From the graph, we can observe that there exist $0<a_1<a_2<a_3<a_4<1$ and $0<b_0<1$ such that \begin{align*}
& g_{17}(a_1,0)=17, \quad g_{17}(a_2, 1)=17, \quad g_{17}(a_4,1) =17;\\
& g_{11}(a_3, 0)=g_{11}(a_3, 1)=11;\\
& g_{17}(1, b_0)=17.
\end{align*}

Also from the graph we know that $g_{11}(a_1,0)\ne 11, g_{11}(a_2, 1)\ne 11$ and $g_{11}(a_4, 1)\ne 11$, so they cannot be equilibrium liquidation strategies. To find out whether $(a_3,0), (a_3,1)$ and $(1, b_0)$ are equilibrium liquidation strategies. We plot the graphs of $g_{11}(a,b_0),  g_{17}(a,0)$ and $g_{17}(a,1)$ as functions of $a\in[0,1]$.

\begin{center}
\includegraphics[scale=0.48]{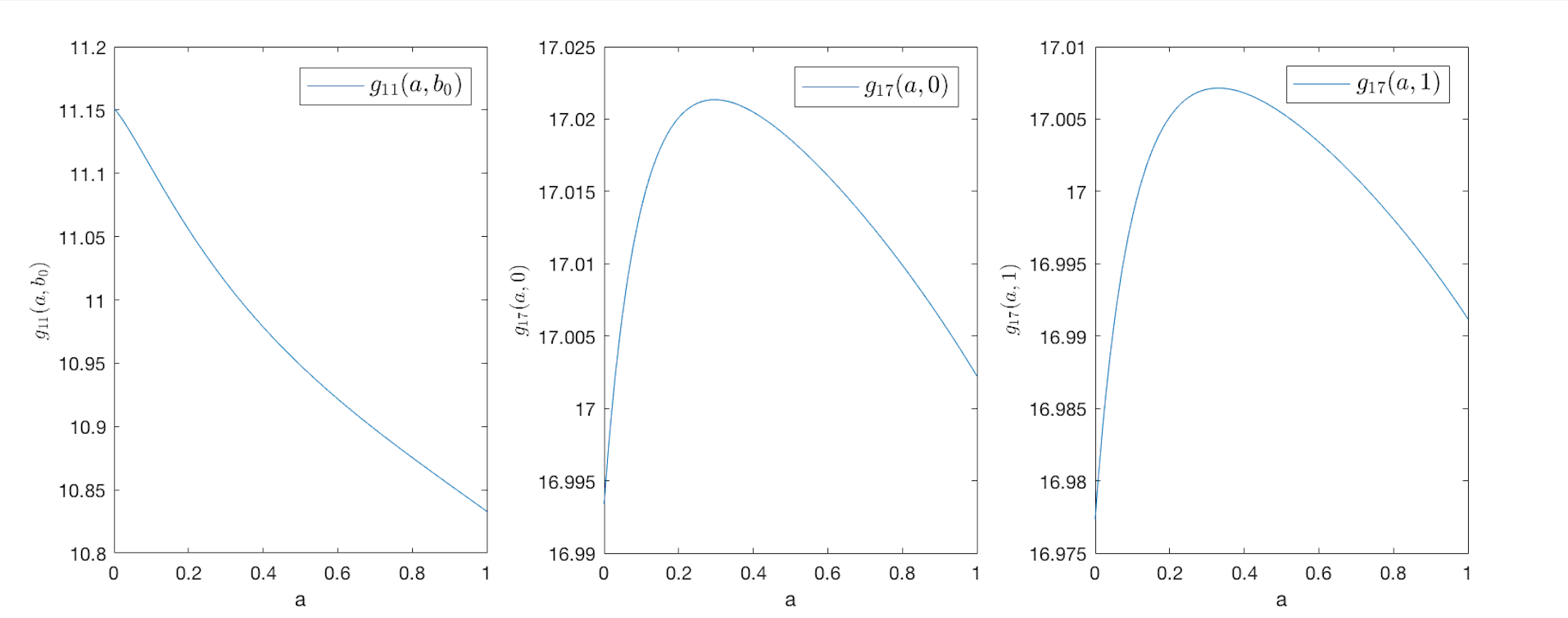}
\end{center}

These graphs show that $g_{11}(1,b_0)<11, g_{17}(a_3, 0)>17$, and $g_{17}(a_3, 1)>17$. So there are five equilibrium liquidation strategies as discussed in part (ii) of the proof of Proposition \ref{examples}. 

\subsection{The third example in Proposition \ref{examples} :} \label{A4}

We first analyze all the possible trajectories of Markov chain $X$ when starting from $1$.

\textbf{Case 1:}
$X:1\to  1\to \cdots \to 1 \to 1\to 0$.

$\Rightarrow Y_{\eta}=0$ for $k=0$ and $Y_{\eta}=a\sum_{i=1}^{k-1}(1-a)^{i}$ for $k\ge 1$. Then $Y_{\eta}=1-(1-a)^k$ with probability $0.1\times0.8^{k}$  for $k\ge 0$.

\textbf{Case 2:}
$X:1\to 1\to \cdots \to 1 \to 1\to 4$.

$\Rightarrow Y_{\eta}=4$ for $k=0$ and $Y_{\eta}=a\sum_{i=1}^{k-1}(1-a)^{i}+4(1-a)^k$ for $k\ge 1$. Then $Y_{\eta}=1+3(1-a)^k$ with probability $0.1\times0.8^{k}$  for $k\ge 0$.

By computation, we have \[
\E_1[Y_{\eta}]=1+\frac{0.2}{0.2+0.8a}, \quad\quad \E_1[Y_{\eta}^2]=1+\frac{0.4}{0.2+0.8a}+\frac{1}{1-0.8(1-a)^2}
\]

Then the explicit expression for $h(a)=\E_1[Y_{\eta}]-c\V_x[Y_{\eta}]$ can be obtained and we have the results in part (iii) of the proof of Proposition \ref{examples}.

\subsection{Equilibrium liquidation strategies for the mean-variance problems in Examples~\ref{eg4.1} and \ref{eg4.2}} \label{A5}

If $\eta$ is an equilibrium liquidation strategy in the mean-variance problem, then \[
J_l(x, {\eta})=\sup_{\xi\in \cL}J_l(x, {\xi\otimes \eta}),\,\,\,\forall x\in \mathbb{X}.
\]
Recall that  $
\E_x[\theta^{\xi\otimes\eta}(X)]=x\xi(x)+(1-\xi(x))\E_x[Y_{\eta}]$ and $\V_x[\theta^{\xi\otimes\eta}(X)]=(1-\xi(x))^2\V_x[Y_{\eta}]$.
Therefore,\[
J_l(x, {\xi\otimes\eta})=-c\V_x[Y_{\eta}]\xi(x)^2+(2c\V_x[Y_{\eta}]-\E_x[Y_{\eta}]+x)\xi(x)+\E_x[Y_{\eta}]-c\V_x[Y_{\eta}],
\]
is a quadratic function of $\xi(x)$ when $\eta$ is fixed. We then have \[
\eta(x)=
\begin{cases}
1, & {\rm if}\,\, h_x(\eta)\in [1,\infty);\\
h_x(\eta), &  {\rm if}\,\, h_x(\eta)\in (0,1);\\
0, & {\rm if}\,\, h_x(\eta)\in (-\infty, 0].\\
\end{cases}
\]
where $h_x(\eta)=\frac{2c\V[Y_{\eta}]-\E_x[Y_{\eta}]+x}{2c\V_x[Y_{\eta}]}$.

In Example \ref{eg4.1}, $\E_i[Y_{\eta}]$ and $\E_i[Y_{\eta}^2]$ for $i=1,7$ have been computed in Appendix \ref{A2}, so we obtain the explicit expressions of $h_1(\eta)$ and $h_7(\eta)$ as functions of $a:=\eta(1)$ and $b=:\eta(7)$. Then we observe that $h_i(a, b)\in (0,1)$, for all $(a,b)\in [0,1]\times[0,1]$, $i=1,7$, and there is exactly one intersection of the curve $\{(a,b): h_1(a, b)=a\}$ and the curve $\{(a,b): h_7(a, b)=7\}$, which is the equilibrium liquidation strategy for mean-variance problem in Example \ref{eg4.1}. Similarly we can find the equilibrium liquidation strategy for mean-variance problem in Example \ref{eg4.2}. The corresponding graphs are shown below.
\begin{center}
\includegraphics[scale=0.28]{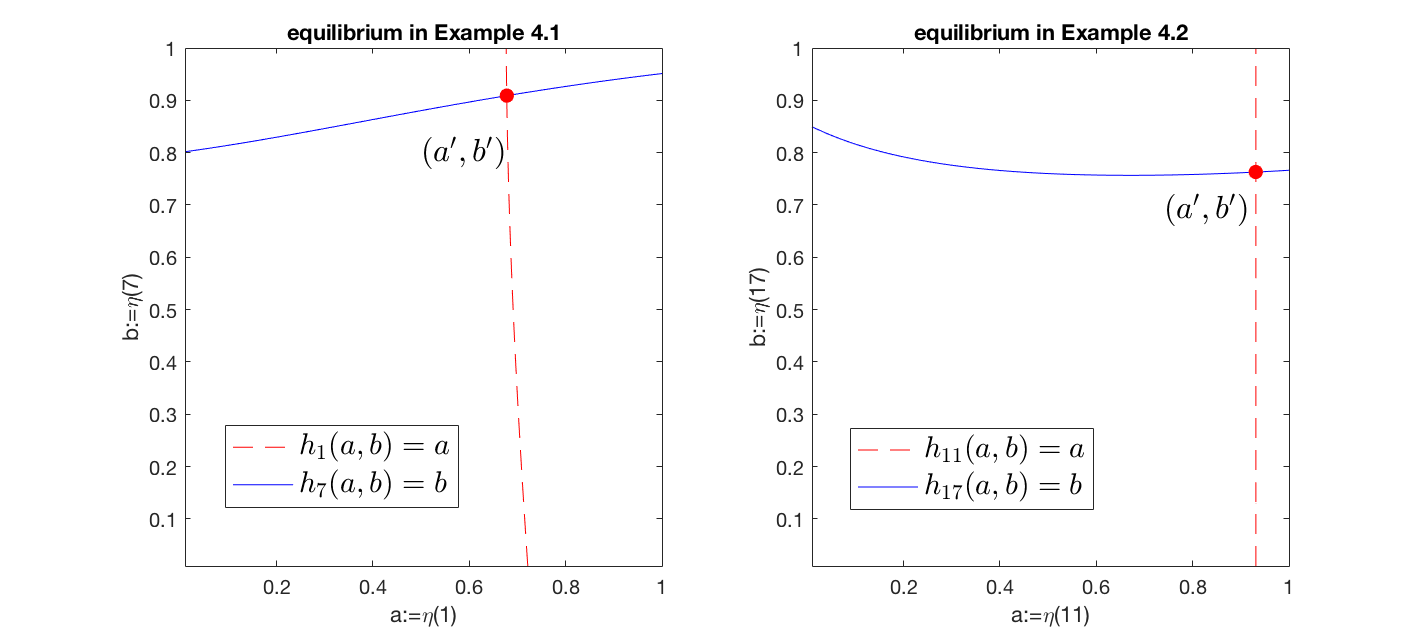}
\end{center}
\bibliographystyle{plain}
\bibliography{reference}{}

\end{document}